\documentclass[12pt]{article}
\usepackage{amsmath}
\usepackage{color,xcolor}
\usepackage{subfiles}
\usepackage{amsfonts}
\usepackage{amsthm}
\usepackage{graphicx}
\usepackage{enumerate}
\usepackage{natbib}
\usepackage{float}
\usepackage{adjustbox}
\usepackage{booktabs}
\usepackage{enumitem}
\usepackage{changepage}
\usepackage{setspace}
\usepackage{soul}
\usepackage{comment}
\usepackage{multirow}

\usepackage[section]{placeins}
\usepackage[colorlinks,linkcolor=blue,anchorcolor=blue,urlcolor=blue,citecolor=blue]{hyperref}%\addbibresource{Bibliography-MM-MC.bib}

\usepackage{multirow}
\usepackage[marginal]{footmisc}

\usepackage{url} % not crucial - just used below for the URL
\usepackage{xr}	
\newtheorem{theorem}{Theorem}%[theorem]
\newtheorem{corollary}{Corollary}%[section]%[theorem]
\newtheorem{lemma}{Lemma}[section]
\newtheorem{assumption}{Assumption}%[section]
  {
	\theoremstyle{plain}
	\newtheorem*{assumption4'}{Assumption 4'}
}
\newtheorem{remark}{Remark}%[section]
\definecolor{red}{RGB}{255,18,0}
\definecolor{yellow}{RGB}{240,228,66}
\definecolor{green}{RGB}{0,158,115}
\definecolor{UAred}{rgb}{0.57, 0.36, 0.51}
\definecolor{lightblue}{rgb}{0.68, 0.85, 0.9}
\definecolor{oldrose}{rgb}{0.75, 0.5, 0.51}
\definecolor{jonquil}{rgb}{0.98, 0.85, 0.37}
\definecolor{UAred}{RGB}{171, 5, 32}
\definecolor{UAblue}{RGB}{21, 36, 74}
\definecolor{Chili}{RGB}{139, 0 , 21}
\definecolor{leaf}{RGB}{112, 184, 101}
\definecolor{AZURITE}{RGB}{30, 82, 136}
\definecolor{MSU}{RGB}{24, 69, 59}
\definecolor{SufeRed}{HTML}{811C21}
\definecolor{UIUCOrange}{HTML}{FF552E}
\definecolor{UIUCBlue}{HTML}{13294B}
\definecolor{DarkCyan}{HTML}{008B8B}
\definecolor{Maroon}{HTML}{8B008B}
\definecolor{SGreen}{HTML}{18453B}

\mathchardef\mhyphen="2D % Define a "math hyphen"

\newcommand{\mF}{\mathcal{F}}

\newcommand{\mR}{\mathcal{R}}

\newcommand{\mS}{\mathcal{S}}
\newcommand{\mP}{\mathcal{P}}

\newcommand{\bbF}{\mathbb{F}}
\newcommand{\bbK}{\mathbb{K}}
\newcommand{\bbQ}{\mathbb{Q}}

\newcommand{\tW}{\widetilde{W}}

\def\E{\textnormal{E}}

\makeatother

\title{Treatment Effects Inference with High-Dimensional Instruments and Control Variables\footnote{We thank Kaicheng Chen, Hidehiko Ichimura, James Powell, Tiemen Woutersen, and participants at the seminar at the University of Arizona for very helpful and constructive comments. Computer programs to replicate the numerical analyses are available from the authors. All remaining errors are our own.}
} 

\author{
    Xiduo Chen\thanks{\small{Department of Economics, University of Arizona, 
    Email: xiduochen@arizona.edu
    }} 
   ,
    Xingdong Feng\thanks{\small School of Statistics and Data Science, Shanghai University of Finance and Economics, Email: feng.xingdong@mail.shufe.edu.cn}
    ,
    Antonio F. Galvao\thanks{\small Department of Economics, Michigan State University, Email: agalvao@msu.edu}
    ,
    Yeheng Ge\thanks{\small Department of Data Science and Artificial Intelligence, The Hong Kong Polytechnic University, Email: yeheng.ge@polyu.edu.hk}
}

\date{\today}

\begin{document}
\maketitle

\begin{spacing}{1.05}
\begin{abstract}
Obtaining valid treatment effect inference remains a challenging problem when dealing with numerous instruments and non-sparse control variables. In this paper, we propose a novel ridge regularization-based instrumental variables method for estimation and inference in the presence of both high-dimensional instrumental variables and high-dimensional control variables. These methods are applicable both with and without sparsity assumptions. To remove the estimation bias, we introduce a two-step procedure employing a ridge regression coupled with data-splitting in the first step, and a ridge style projection matrix with a simple least squares regression in the second. We establish statistical properties of the estimator, including consistency and asymptotic normality. Furthermore, we develop practical statistical inference procedures by providing a consistent estimator for the asymptotic variance of the estimator. The finite sample performance of the proposed methods is evaluated through numerical simulations. Results indicate that the new estimator consistently outperforms existing sparsity-based approaches across various settings, offering valuable insights for complex scenarios. Finally, we provide an empirical application estimating the causal effect of schooling on earnings addressing potential endogeneity through the use of high-dimensional instrumental variables and high-dimensional covariates.

%\Kchange{
%This paper proposes novel ridge regularization-based methods for estimating treatment effects and conducting inference in the presence of both high-dimensional instrumental variables and high-dimensional control variables. These methods are applicable with or without sparsity assumptions. To mitigate the high-dimensional instrument biases, we propose a novel two-step estimator that employs a ridge regularization to the control variables coupled with a data-splitting strategy. We establish the statistical properties of the estimator, namely, consistency and asymptotic normality. In addition, we develop statistical inference procedures by suggesting a consistent estimator for the asymptotic variance of the estimator. We study the finite sample performance of the methods using numerical simulations. Results document evidence that the new estimator outperforms existing sparsity-based approaches across a variety of settings.
%}

\end{abstract}

\vspace{2mm}

\end{spacing}

\doublespacing
\newpage
\section{Introduction}

%There is a large and growing literature on program evaluation studies. 

Estimation of causal and treatment effects has provided a valuable method of statistical analysis and understanding of policy variable effects. This is especially true for program evaluation studies in economics, finance, and statistics, where these methods help to analyze how treatments or social programs affect the outcome variables of interest. In addition, with the increasing availability of richer data sets, the use of high-dimension methods in treatment effects models has become widespread in the literature (see, e.g., \citet{BelloniChernozhukovHansen14}, \citet{chernozhukov2015post}, \citet{chernozhukov2018debiased}, and 
\citet{angrist2022machine}). 

Endogeneity of the variable of interest is a very pervasive issue in empirical applications, and one of the most important challenges faced by researchers. The instrumental variables (IV) analysis is commonly used in practice to compute treatment effects for endogenous regressors (see, e.g., \citet{ImbensRubin15} and references therein). This is notably relevant for estimating causal effects when instruments are exogenous conditional on observables.
Recently, the use of high-dimensional linear IV models has proven to be crucial for correct and accurate estimation and inference in these studies (see, e.g., \citet{belloni2012sparse}, \citet{hansen2014instrumental}, \citet{zhu2018sparse}, \citet{CattaneoJanssonMa19}, \citet{gold2020inference}, and \citet{QiuTaoZhou21}).
This scenario is prevalent when the number of variables is very large, or when interactions between variables should be considered. 
It has been common in the literature to separately allow for high-dimensional IV or high-dimensional control variables.\footnote{There is also an existing literature considering the case of many IV. In this case, it is common that the number of IV grows together with the sample size, at the same rate or slower, but it is not allowed to be larger than the sample size (see, e.g., \cite{hausman2012instrumental}, \cite{ChaoSwansonWoutersen23}, and \cite{mikusheva2020inference}). In the high-dimension IV literature, the number of instruments is allowed to be larger than the sample size.} Ignoring the high dimension may result in biased inferences and misinterpretation of causal effects.

%\textcolor{red}{ "Inference for high-dimensional exchangeable arrays" Harold D. Chiang, Kengo Kato, and Yuya Sasaki Journal of the American Statistical Association, 118 (543), pp. 1595-1605 (2023)}

A typical solution to the high-dimension problem is to assume variable selection under the (approximate) sparsity, which effectively imposes that virtually all control coefficients are (very near) zero, and only a limited number of variables have substantial magnitude effects on the response. Many penalized-based regression methods have been proposed to obtain consistent estimators and confidence intervals for scalar parameters under the approximate sparsity assumption, including literature on selecting among high-dimension controls (see, e.g., \citet{zhang2014confidence}, \citet{farrell2015robust}, \citet{chernozhukov2015post}, \citet{belloni2015uniform}, and \citet{zhu2018sparse}) and also in the high-dimension IV context (see, e.g., \citet{belloni2012sparse}, and \citet{belloni2016post}).   
Although the sparsity assumption is common for estimation and inference in many models in the presence of high-dimensional instruments and controls, it is unverifiable and may often be violated in empirical applications.
%For example, if we want to include many fixed effects, it is not always clear why the vast majority of the fixed effect coefficients should be close to zero.
Once sparsity is violated, the estimator may be severely biased, and the corresponding confidence interval for the low dimensional parameter is inconsistent. Thus, considering inference under a non-sparsity assumption is of practical importance. 

Non-sparse structures of control variables have been studied in the literature. Recently, \citet{cha2023inference} provide a general framework for inference in high-dimensional regression models without assuming sparsity, incorporating multiple estimation steps to account for complex dependence structures.
A portion of it weakens the sparsity condition (see, e.g., \citet{he2000parameters}, \citet{buhlmann2003boosting}, \citet{belloni2014inference}, \citet{hansen2014instrumental}, \citet{kozbur2017testing}, \citet{solvsten2020robust}, \citet{kozbur2020analysis}, \citet{mikusheva2020inference}, \citet{ing2020model}, and \citet{li2021LimitControl}). 
A compelling strategy to address the high-dimensional variable problem in the absence of sparsity is to employ a two-step ridge regression. \citet{hansen2014instrumental} use this framework and suggest a two-stage ridge jackknife instrumental variables estimator (RJIVE) under high-dimensional instruments. However, also considering high-dimensional control variables in this framework can add further challenges, since correlations between the two stages cannot be estimated/eliminated by conventional procedures.

In this paper, we contribute to the literature on IV models by proposing estimation and inference of treatment effects in the presence of both high-dimensional instruments and high-dimensional control variables with a non-sparse structure.  We first show that, in the presence of high-dimensional control variables, unfortunately, the existing RJIVE estimator suffers from inconsistency, which is induced by the use of the jackknife strategy together with the ridge regression when controlling for the high-dimensional exogenous controls in the first stage. Hence, we introduce a novel two-step estimator that employs a ridge regularization to the instrumental variables coupled with a data-splitting strategy, as well as a ridge projection matrix to partially out exogenous regressors.\footnote{\cite{angrist1995split} use a sample splitting strategy to solve the two-stage least squares bias. \cite{belloni2012sparse} employ sample splitting in the context of sparse models.} 
In contrast to existing approaches, which often involve iterative model selection, boosting techniques, or multi-step procedures, our proposed method uses a two-step estimator, where the first stage uses ridge regression coupled with data-splitting, and the second stage uses a ridge-style projection matrix with simple least squares regression.
%\textcolor{red}{Our method uses a two-step ridge estimator with data splitting and a ridge-style projection, providing a simpler and more efficient alternative.} 
This framework offers a simpler and computationally efficient alternative, and  specifically addresses challenges in models with both high-dimensional IV and exogenous controls, introducing new features related to instrument-control correlations that our method is designed to handle.

The practical implementation of the estimator is simple. 
In particular, in the first step, one uses the first partition of the data to compute a ridge regression of the endogenous variable on the set of instruments. Then, given these estimates, one uses the second partition of the data to calculate the fitted values.  In the second step, to compute the treatment effects estimator, one uses a projection matrix with a ridge penalty term to partial out the high-dimensional exogenous covariates together with a simple ordinary least squares estimator regressing fitted values from the first step on the outcome variable. Mild sufficient conditions are provided for this two-step estimator to have the desired asymptotic properties, namely, consistency and asymptotic normality. We highlight that this new framework is valid with or without sparsity assumptions. By avoiding variable selection procedures, the proposed methods are more robust than sparsity-based methods for causal inference with both high-dimensional IV and controls.

%The ridge regularization approach is not an entirely novel method. If we consider solely the instrumental variable component, it aligns with the technique proposed by Hansen's paper, albeit with a slightly modified jackknife procedure that does not significantly differ. The key distinction lies in our application of ridge regularization to the control variables, coupled with a data-splitting strategy. Furthermore, we employ a ridge-type variance estimation technique, which offers improved efficiency compared to traditional methods.

We develop statistical inference procedures for the proposed methods. It has been known that when the covariate dimension is large relative to the sample size, many classical methods become invalid for inference. \cite{liu2020estimation} develop inference methods based on ridge regression that can be applied to both low- and high-dimensional models. We extend their results to our two-stage treatment effects estimator. In particular, we employ a ridge-type variance estimation technique, which offers improved efficiency compared to traditional methods, and is easy to compute. We formally establish the consistency of the variance estimator. Finally, since the parameter of interest is finite-dimensional and the weak limit of the estimator is standard, one is able to employ standard inference procedures.  

We conduct Monte Carlo simulations to evaluate the finite sample performance of the proposed methods.\footnote{The practical implementation of the proposed method is available in the \texttt{R} package ``HDRRTreat''.} We investigate models with and without the sparsity condition. Numerical simulation results document evidence that the new estimator outperforms existing sparsity-based as well as non-sparsity-based approaches across a variety of settings. The proposed estimator is approximately unbiased, and the empirical coverage of the confidence interval is close to the nominal one. Overall, simulation results are in line with the theory, and provide favorable numerical evidence for the suggested methods.

Finally, we provide an empirical application to illustrate the methods proposed in this paper. We revisit the classic example of \citet{angrist1991does} and estimate the causal effect of schooling on earnings by addressing potential endogeneity through the use of high-dimension instrumental variables and high-dimension covariates. Results document evidence of imprecise estimates, with large confidence intervals.

The remainder of the paper is organized as follows. Section \ref{sec:model} presents the model and the parameter of interest. In Section \ref{sec:estimator}, we discuss the two-step estimation procedure. The limiting statistical properties of the estimator are presented in Section \ref{sec:asymp theory}. Section \ref{sec:inference} presents inference procedures, and Section \ref{sec:simulation} collects numerical Monte Carlo results. An empirical illustration of the methods is provided in Section \ref{sec:application}. Finally, Section \ref{sec:conlusion} concludes. All proofs are collected in the Appendix.

\section{Model}\label{sec:model}

The main objective of this paper is to produce causal inference in a model that allows for endogeneity of the variable of interest, as well as high-dimensional instruments and high-dimensional exogenous covariates. The study of causal effects is an important tool for practitioners. 
We consider the following model
\begin{align}
\label{reg_main_general}
y_{i}&
= \alpha d_i + f(X_{i})+\epsilon_{i} 
% = \alpha d_i +\gamma_{x}^{\top}X_{i}+\epsilon_{i} 
\\
\label{reg_iv_general}
d_{i}&=\Upsilon_i(Z_{i}) +  v_{i},
% = \gamma_{z}^{\top}Z_{i}+v_{i},
\end{align}
where $y_i$ is the response variable, $d_i $ is a scalar endogenous treatment variable,
% Without loss of generality, we set $d_i$ as a scalar variable, and denote $ \bd = (d_1,\dots, d_n)^{\top}$.
and $\alpha$ is the main treatment effects parameter of interest. The vectors $X_i$ and $Z_{1i}$ contain the exogenous controls and the instruments variables, respectively. We consider $X_i$ and $Z_{i}:=[Z_{1i},X_{i}]$ to have dimensions $p_x$ and $p_z$, respectively, with $p_{z} = p_{x} + p_{z_{1}}$.
The function $\Upsilon_i$ is the unknown optimal instrument, and $f(\cdot)$ is the unknown function of exogenous controls.
% The instruments affect the treatment variable via equation \eqref{reg_iv}, and the confounding variables via equation \eqref{reg_main}. 
The endogeneity in the above system comes from the correlation of errors $\epsilon_i$ and $v_i$. 

It is common in applications to specify and estimate simple linear models. We follow the literature and assume a linear specification for the controls $f(\cdot)$ and the optimal instruments $\Upsilon(\cdot)$. Although both of these functions might be unknown, and potentially complicated, the linear specification can be satisfied with a series expansion of the observable variables.\footnote{For example, one could have $X_{i} =\{p_{kK}(\zeta_{i})\}_{k=1}^{K}$ and $Z_{i} =\{p_{mM}(\upsilon_{i})\}_{m=1}^{M}$ for some set of basis functions $\{p_{jJ}(\cdot)\}$ such as orthogonal polynomials or splines formed from, respectively, a set covariates $\zeta_{i}$, and a set of `fundamental' instruments $\upsilon_{i}$.} Hence, we consider the following model
\begin{align}
y_{i} & = \alpha d_i +\gamma_{x}^{\top}X_{i}+\epsilon_{i} \label{reg_main},\\
d_{i} & =\gamma_{z}^{\top}Z_{i}+ v_{i}    \label{reg_iv}, 
\end{align}
with $Z_i=[Z_{1i},X_{i}]$. We assume that $\E[v_{i}|Z_{i}]=0$ and that $\E[\epsilon_{i}|Z_{i}]=0$ for all $i$. Further, we assume that $\textnormal{Var}(Z_{i})\neq 0$ such that $Z$ is a valid instrument for $X$. 

In this work, we allow for both the number of exogenous covariates and available instruments to be larger than the sample size. In particular, we assume that $p_x/n \rightarrow \tau_x$ and $p_z/n = (p_{z_{1}}+p_x)/n  = (p_{z_{1}}/n+p_x/n)=\tau_{z_{1}}+\tau_x \rightarrow \tau_z$, where $\tau_{x}>0$ and $\tau_{z}>0$. 
The parameter of interest, $\alpha$, is presented here as a scalar for simplicity. %Note that the extension to a vector-valued $\alpha$ is straightforward. 
The methods and results discussed can be readily generalized to accommodate a vector-valued $\alpha$, preserving the same theoretical properties and interpretations.
Next, we introduce an estimator for the causal parameter of interest $\alpha$.

%\XDchange{I am changing it to be scalar, and also for the $d$ - It's a column vector now. I will comment out this after all is finished}

%\XDchange{Change the notation to BX, instead of $\gamma$, in line with the notation afterward.}

\section{The estimator}\label{sec:estimator}

This section proposes a two-step estimator for the treatment effects parameter of interest. Estimating the parameters of the model in equations \eqref{reg_main}--\eqref{reg_iv} in a high-dimensional non-sparse setting presents several challenges. First, it is computationally demanding to handle large numbers of instruments and exogenous control variables concurrently. Second, while sparsity generally facilitates the derivation of statistical properties for a penalized estimator, such as Lasso (see, e.g., \cite{chernozhukov2015post}), it comes at the cost of complicated variance term derivations. Finally, and most crucially, simply using the ridge jackknife instrumental variables estimator (RJIVE) --- that is often used in the high-dimensional instrumental variable context (see, e.g., \citet{hansen2014instrumental}) --- together with standard partialling out the exogenous regressors induces bias. 

Before we describe the details of our proposed method, to fix ideas and build intuition, we discuss reasons why a simple extension of the existing two-stage RJIVE method would be inadequate for estimating a model with endogeneity and high-dimensional regressors and instruments.

Let us consider a simple potential extension of the RJIVE estimation procedure for the model in equations \eqref{reg_main}--\eqref{reg_iv} above, which is inspired by a ridge regularization extension of the JIVE idea to the $n>p$ case. In the first-stage regression, the optimal instrument $\widehat{d}$ is estimated with a ridge regularized jackknife method as follows,
\begin{equation}\label{eq:opt instrument}
\widehat{d}_i =  \widehat{\gamma}_z(i)Z_{i}=\left(d_{-i}^{\top} Z_{-i}\right)\left(Z_{-i}^{\top} Z_{-i} + \lambda I \right)^{-1}Z_{i},
\end{equation}
where $\lambda$ is the regularization parameter, and the subsamples $d_{-i}$ and $Z_{-i}=[Z_{-1i}, X_{-i}]$ consist of all but the $i$-th data point, and $I$ is a $p_{z} \times p_{z}$ identity matrix.
This jackknife strategy introduced by \citet{angrist1999jackknife} is able to eliminate the many instrument biases for an independent sample. In particular, it has been shown that, under the exogeneity of the instrument, and a small number of covariates, the JIVE strategy does not suffer from the many instrument biases, because of the following orthogonality
\begin{align}\label{eq:jive}
\E\left[ \epsilon_i  \widehat{d}_{i}  \right]
& = \E\left[ Z_i^\top \left(
Z_{-i}^\top Z_{-i} + \lambda I
\right)^{-1} \left(Z_{-i}^\top \right) \E\left[d_{-i} \epsilon_i \mid Z_{-i}\right]
\right] \nonumber\\
&= \E\left[ Z_i^\top \left( Z_{-i}^{\top} Z_{-i} + \lambda I \right)^{-1} Z_{-i}^{\top} 
\left(
\E\left[Z_{-i}\gamma_z \epsilon_i \mid Z_{-i}\right]
+ \E\left[v_{-i} \epsilon_i \mid Z_{-i}\right]
\right)
\right]\nonumber\\
&=0.    
\end{align}
The last equality in \eqref{eq:jive} is zero because both $\E\left[v_{-i} \epsilon_i \mid Z_{-i}\right]=0$ and $\E\left[Z_{-i}\gamma_z \epsilon_i \mid Z_{-i}\right]=0$. The former conditional expectation is zero due to $\epsilon_{i}$ being independent of $v_{j}$ if $j\neq i$, and the latter is zero due to the validity of the instrument. 

Unfortunately, this jackknife strategy cannot accommodate scenarios with an additional large number of covariates, since equation \eqref{eq:jive} will not be valid. In such high-dimensional regressor setting, a simple partialling out of exogenous variables is not feasible, and the resulting bias can no longer be disregarded. To see this, consider the residuals from equation \eqref{reg_main} after partialling out covariates using ridge regression:
\begin{equation*}
\widehat{\epsilon} = \frac{1}{1-n^{-1}tr(A_{n})}n^{-1}(I-A_{n})\epsilon, 
\end{equation*}
where $A_{n}=X(X^{\top}X+\eta I)^{-1} X^{\top}$. The matrix $A_n$ denotes a feasible ridge projection matrix, regularized with the penalized term $\eta I$ (see, e.g., \citet{hansen2014instrumental} for further discussion on the properties of this matrix). 

Now, recall the optimal instrument in equation \eqref{eq:opt instrument}, by taking the expectation of term $\widehat{\epsilon}_i \widehat{d}_i$, we obtain
\begin{align}\label{eq:jive2}
\E\left[ \widehat{\epsilon}_i \widehat{d}_i \right]
& = \E\left[
Z_i^{\top}
\left(
Z_{-i}^{\top} Z_{-i} + \lambda I
\right)^{-1}
\left(Z_{-i}^{\top}\right) \E\left[d_{-i} \widehat{\epsilon}_i \mid Z_{-i}\right]
\right]\nonumber\\
& =
\E\left[
Z_i^{\top}
\left(
Z_{-i}^{\top} Z_{-i} + \lambda I
\right)^{-1}
\left(Z_{-i}^{\top}\right) 
\left(
\E\left[Z_{-i}\gamma_z \widehat{\epsilon}_i \mid Z_{-i}\right]
+
\E\left[v_{-i} \widehat{\epsilon}_i \mid Z_{-i}\right]
\right)
\right]\nonumber\\
& \neq 0.
\end{align}
It is important to note that the term $\E\left[v_{-i} \widehat{\epsilon}_i \mid Z_{-i}\right]$, in the above expression, is no longer equal to zero due to the fact that $\widehat{\epsilon}_i$ contains information from the entire sample because the matrix $A_n$ has information about every $i$, including the observation $-i$. Hence, $\widehat{\epsilon}_i$ is related to $v_{-i}$. 
%\textcolor{red}{Intuitively by jackknife, $\epsilon_i$ and $\widehat{v}_{-i}$ are independent but we cannot clear out all the $v_i$ in $\epsilon_i$ because $A_n$ has information about $i$. To be precise, $v$ and $\epsilon$ are NOT orthogonal and that's where the endogeneity comes in. BUT $v_{-i}$ should be independent of $\epsilon_i$ because 1.ideally, the correlation of $\epsilon_i$ and $v_i$ are captured by $Z_i$ 2. $v_i$ and $v_{-i}$ are independent. }

Consequently, the lack of orthogonality in \eqref{eq:jive2} induces inconsistency of this JIVE estimator.
Intuitively, the inconsistency arises because the JIVE estimator relies on the assumption of independence between the error term and the constructed instrument, an assumption that is violated when the error term incorporates information from the entire sample.
Thus, the JIVE strategy is not designed to accommodate the case with both high-dimensional instrumental variables and high-dimensional control variables. 

To address the shortcomings of simultaneously accounting for a large number of instruments and control variables, we propose a novel two-step ridge regression approach for estimation under a non-sparsity assumption. To handle the high-dimension instrument and controls concurrently, we make use of two tools that are usually available in the literature. For estimation, we incorporate a data-splitting method in conjunction with a regularized (semi-)projection matrix for partialling out the exogenous regressors. 

The two-step ridge regression estimation procedure is as follows. In the first step, we start by splitting the data set into two parts randomly. %Without loss of generality, we assume sample size $n$ is an even number. 
Denote the partitions of the data $\{1,\cdots,n\}$ into two subsets $\mS_1 = \{1,\cdots,n_1\}$, $\mS_2 = \{1,\cdots,n_2\}$ and $n=n_1+n_2$.
Then, regarding the first-stage regression in equation \eqref{reg_iv}, we use the first partition of the data, and consider the following ridge regression estimation for the parameter $\gamma_{z}$:
\begin{equation*}
\widehat{\gamma}_{z}=\arg\min_{\gamma_{z}} \sum_{i \in \mS_1} \left( d_{i}-Z_{i}^{\top}\gamma_{z} \right)^2 + \eta_{z}\| \gamma_{z}\|_{2}^{2},     
\end{equation*}
where $\| \cdot \|_{2}$ is the Euclidean norm, and $\eta_{z}$ is a tuning parameter that we will discuss below. Thus, the estimator can be written as follows:
\begin{align}\label{eq: FS estimator} 
\widehat{\gamma}_{z}=\left(Z^{\left(1\right)\top}Z^{\left(1\right)}+\eta_{z}I\right)^{-1}\left(Z^{\left(1\right)}d^{\left(1\right)}\right),
\end{align}
where $Z^{\left(1\right)}$ and $d^{\left(1\right)}$ denote the instruments and endogenous variables, respectively, with sub-samples indexed by $\mS_1$. 
Now, using the second partition of the sample, together with the above estimates, we construct the fitted values as,
\begin{align}\label{eq: FS dhat}  \widehat{d}^{(2)}=Z^{(2)\top}\widehat{\gamma}_{z},
\end{align}
where $Z^{\left(2\right)}$ denotes the instruments for the sub-sample indexed by $\mS_s$. 

In the second step, we estimate the parameter $\alpha$ in equation \eqref{reg_main} by performing a simple ordinary least squares procedure after partialling out regressors as follows:
\begin{align}\label{eq: estimator}    \widehat{\alpha}=\left(\widehat{d}^{(2)\top}\left(I-A_{n_{2}}\right)d^{(2)}\right)^{-1}\left(\widehat{d}^{(2)\top}\left(I-A_{n_{2}}\right)y^{(2)}\right),
\end{align}
where 
\begin{equation*}
A_{n_{2}}=
n_{2}^{-1}
X^{(2)}\left(X^{(2)\top}X^{(2)}+\eta_{x}I\right)^{-1}X^{(2)\top}.
\end{equation*}
It is important to note that the notion of partialling out above uses the matrix $A_{n_{2}}$, which is a regularized (semi-)projection matrix with $\eta_{x}$ being a tuning parameter that we discuss further below.

We remark that, unlike \cite{hansen2014instrumental}, our proposal incorporates a penalized term --- through the $A_{n_{2}}$ projection matrix --- in the second stage of estimation, rather than directly using the estimated values from the first stage. This approach enhances the robustness and stability of our estimator by regularizing the second stage, addressing high-dimensional issues in both stages.
%\XDchange{They didn't use the panelized term in the second stage, but directly plug in the estimated value from the first-stage - will double check} \textcolor{red}{(Correct. Their estimator is in equation (3.5), page 294.)}
The two-step estimator in equation \eqref{eq: estimator} has two important features. First, the sample-splitting procedure allows for consistent estimation of the parameter of interest and is crucial for inference purposes.
%First, the sample splitting procedure, which allows for consistent estimation of the parameter of interest, and is important for inference purposes. 
Technically, the procedure in equations \eqref{eq: FS estimator} and \eqref{eq: FS dhat} guarantees the independence of two leading terms in the linear representation of the estimator in the first step of estimation --- one term is the projection of the endogenous regressor on the set of instruments, and the other is the projection of the outcome variable on the set of exogenous controls. Hence, the sample splitting ensures that the limiting distribution of the estimator is well-behaved. Intuitively, this strategy allows for correct inference because the randomness inherent in the estimated treatment variable is independent of the dataset used for statistical inference. The data split strategy can also resolve the many instrument biases in the presence of controls due to the independent nature of the estimation and inference dataset.

%\textcolor{red}{COMMENT on SAMPLE SPLITTING: Hansen's paper argues that low-dimensional controls can be partialled out without affecting subsequent JIVE modeling. However, after partialling out the controls, we obtain $\widehat{d} = (I - P_X)v + (I - P_X)Zb$. Consequently, even in a low-dimensional setting, the term $(I - P_X)v$ is not i.i.d., and their proof does not account for the endogeneity bias introduced by this term. So precisely speaking, the procedure that everyone uses MAY yield some problems with iid assumption. In low dimensional case and asymptotically, it might be fine, but in our DGP it won't work well. I THINK we should mention this but may need to carefully think about how to say it clearly \& gently.} 
The second feature is the partialling out of the high-dimensional control variables by the ridge projection matrix, $A_{n_{2}}$. As discussed previously, the matrix $A_{n_{2}}$ is a regularized (semi-) projection matrix, and the practical implementation of the estimator requires the selection of the penalty parameters $\eta_{x}$ and $\eta_{z}$. %\textcolor{red}{We follow the usual approach for ridge regression literature (see, e.g., \citet{hansen2014instrumental}, \citet{carrasco2012regularization}) and set $\eta_{z}=C_{z} p_{z}$ and $\eta_{x}=C_{x} p_{x}$, where $C_{z}$ and $C_{x}$ are constants of proportionality, and $p_{z}$ and $p_{x}$ are the dimensions of the instrumental and controls, respectively.} \textcolor{orange}{(FIX THIS Following SENTENCE.)} \textcolor{red}{We set the constants of proportionality $C_{z}$ and $C_{x}$ for each first-stage equation to the sample standard deviation of the element of $X_i$ being considered after partialling out any included control variables.}
We follow the usual approach in the ridge regression literature (see, e.g., \citet{friedman2010regularization}, and \citet{liu2020estimation}) and set $\eta_{z}=c_{z} \max_{1 \leq j \leq p} \big| z_{(j)}^{\top} d \big| / (n_{1} p_{z})$, and $\eta_{x}=c_{x} \max_{1 \leq j \leq p} \big| x_{(j)}^{\top} Y \big| / (n_{2} p_{x})$, where $c_{z}$ and $c_{x}$ are preset constants. Here,  $Y = (y_{1}, \ldots, y_{n})^{\top}$, $x_{(j)} = (x_{1j}, \ldots, x_{nj})^{\top}$ is the $j$-th column of $X \ (j = 1, \ldots, p_{x})$, and $z_{(j)} = (z_{1j}$, \ldots, $z_{nj})^{\top}$ is the $j$-th column of $Z \ (j = 1, \ldots, p_{z})$.

The tuning parameters $\eta_{z}$ and $\eta_{x}$ are used in the first and second stages, in equations \eqref{eq: FS estimator} and \eqref{eq: estimator}, respectively. Parameters $\eta_{z}$ and $\eta_{x}$ capture the scale of the data, ensuring that the regularization terms are appropriately scaled relative to the dimensions of the control variables and instrumental variables, respectively. Specifically, they are designed to balance the regularization based on the maximum influence of individual predictors on the response variable $Y$ and the dependent variable $d$. In practice, we use $\eta = \min\{\eta_{x}, \eta_{z}\}$ for regularization to achieve a balanced approach.

In our proposed estimator, the tuning parameters control the extent of shrinkage applied to the coefficients rather than directly determining the sparsity of the model. In contrast to the $L_1$ norm requirements in \cite{chernozhukov2015post}, where the tuning parameter is used to achieve varying levels of sparsity or to eliminate predictors in model selection, our method focuses on regularizing the magnitude of the coefficients, allowing them to be small but non-zero. It is worth noting that in this context, by employing ridge-regularization together with sample splitting, we can effectively control the complexity of the model while maintaining the inclusion of all predictors, thereby striking a balance between model interpretability and predictive performance.

A related contribution is \citet{chernozhukov2018debiased}, who develop a debiased machine learning approach combining orthogonalization with sample splitting to deliver valid inference in high-dimensional settings. Our method differs in that it emphasizes shrinkage-based regularization rather than debiasing, aiming at stable estimation while retaining all predictors, and is thus, although related, more directly comparable to ridge-type procedures.

The next section shows that the two-step procedure discussed here has desirable statistical properties. In particular, we establish the consistency and asymptotic normality of the estimator. We will also discuss practical inference.

\section{Theoretical results}\label{sec:asymp theory}

This section presents the statistical properties of the two-step estimator discussed in the previous section. We first discuss the assumptions, and then establish the consistency and asymptotic normality of the estimator.

Notation is as follows: We have a dataset $\{X_{i},Z_{i},y_{i},d_{i}\}_{i=1}^{n}$ of size $n$, where $d_{i}$ is the treatment, $X_{i}$ are control variables, $Z_{i}$ are high-dimensional instrument variables, and $y_{i}$ is the outcome variable. We randomly split the dataset into two parts, $\mS_{1}$
and $\mS_{2}$, with size $n=n_1+n_2$. We denote the dataset $\mS_{1}$ for the first partition $\{X_{i}^{\left(1\right)},Z_{i}^{\left(1\right)},y_{i}^{\left(1\right)},d_{i}^{\left(1\right)}\}_{i=1}^{n_1}$, and $\mS_{2}$ for the second partition $\{X_{i}^{(2)},Z_{i}^{(2)},y_{i}^{(2)},d_{i}^{(2)}\}_{i=1}^{n_2}$. Finally, our goal is to estimate and conduct inference for the treatment effect parameter $\alpha$.

\subsection{Assumptions}
%\hl{Rearranged to only contain the assumption on data, and I will put the assumptions on $eta$ into the theorem}.

We start by stating and discussing the assumptions to establish the statistical properties of the estimator presented above. 
Let $\{X_{i},Z_{i},y_{i},d_{i}\}_{i=1}^{n}$ be  independent and identically distributed (\textit{i.i.d.}) observations. We also consider the following conditions.

\begin{assumption}
    \label{assumption_Eigen_Upsilon}
     The optimal instrument can be written as
    $\Upsilon_i = \mP_X\Upsilon_i + (I-\mP_X)\Upsilon_i = \vartheta_{X,i} + \vartheta_{X^c,i}$, where $\mP_X = X (X^\top X)^{-1} X^\top$ is the projection matrix onto the  space spanned by $X$.  
    % for some $\vartheta_i$ and $S_n$ where 
    % $S_n = \widetilde{S}_n diag(\mu_{1n},\dots,\mu_{Gn})$,
    % $\widetilde{S}_n$ is $G \times G$ and bounded and the smallest eigenvalue of $\widetilde{S}_n\widetilde{S}_n^{\top}$ is bounded away from zero.
    Also, there is $C>0$ such that 
    $\left\|\sum_{i=1}^{n} \vartheta_{X^c,i}\vartheta_{X^c,i}^{\top} / n\right\| \leq C$ where $\| \cdot \|$ denotes the Euclidean norm, and the smallest eigenvalue $\lambda_{\min }\left(\sum_{i=1}^{n} \vartheta_{x^c,i} \vartheta_{x^c,i}^{\top} / n\right) \geq 1 / C$. In addition, the following moment conditions hold:
    % \begin{equation*}
    $$
    \E[\, v_i \mid Z_i \,] = 0,
    \qquad
    \E[\, \epsilon_i \mid Z_i \,] = 0,
    $$
    % \end{equation*}
    and the instrument satisfies the relevance condition
    % \begin{equation*}
    $
    \mathrm{Var}(Z_i) \neq 0.$
    % \end{equation*}
    Equivalently, $\E[\, \vartheta_{X^c,i} \, u_i \,] = 0$,
    where $u_i$ denotes the structural error term. 
\end{assumption}

Assumption \ref{assumption_Eigen_Upsilon} describes the asymptotic representation property of optimal instruments. This condition decomposes the optimal instruments into two components: one residing in the column space of confounding variables, and another orthogonal to the confounding variables. This ensures stability and control over these terms as the sample size increases. Assumption \ref{assumption_Eigen_Upsilon} is extensively employed and documented in the literature on many instruments (see, e.g., \citet{hausman2012instrumental} and \citet{hansen2014instrumental}).

 % \begin{assumption} %\label{assumption_reg}
 %  \textcolor{orange}{(HERE FOR COMPARISON. TO BE DELETED.)}
 %     \begin{itemize}
 %         %\item  $\|n^{-1}\eta B^{\top}B\|=o(1)$
 %         \item[(i)] $\lambda_{\max}\left(\frac{X^{\top}X}{n}\right)<C,\quad\lambda_{\min}\left(\frac{X^{\top}X}{n}\right)>c$, $\|\gamma_{x}\| < C$ for some constants $C,c > 0$;
 %         \item[(ii)] $\lambda_{\max}\left(\frac{Z^{\top}Z}{n}\right)<C,\quad\lambda_{\min}\left(\frac{Z^{\top}Z}{n}\right)>c$,
 %         $\|\gamma_{z}\|<C$
 % for some constants $C,c> 0$.
 %     \end{itemize}
 % \end{assumption}

%This condition constitutes a relevance of the instrument condition. It stipulates that the ridge-predicted value of the optimal instrument exhibits a correlation with the unobserved optimal instrument. 

\begin{assumption} \label{assumption_reg}
Let $\Sigma_X = \E(X^{\top}X)$ and $\Sigma_Z = \E(Z^{\top}Z)$ be the  covariance matrix of controls and instruments, respectively. There exist constants $C > 0$ and $c > 0$ such that, as $n \to \infty$,
\begin{itemize}
    \item[(i)] $\lambda_{\max}\!\left(
    % \frac{X^{\top}X}{n}
    \Sigma_X
    \right) < C$ and $\lambda_{\min}\!\left(
    % \frac{X^{\top}X}{n}
    \Sigma_X
    \right) > c$
    % with probability approaching one (w.p.a.1)
    , and $\|\gamma_x\| < C$. 
    \item[(ii)] $\lambda_{\max}\!\left(
    % \frac{Z^{\top}Z}{n}
    \Sigma_Z
    \right) < C$ and $\lambda_{\min}\!\left(
    % \frac{Z^{\top}Z}{n}
    \Sigma_Z
    \right) > c$ w.p.a.1, and $\|\gamma_z\| < C$. 
\end{itemize}
\end{assumption}

Assumption \ref{assumption_reg} ensures that the signal retained after regularization constitutes a non-negligible fraction of the original signal, in line with the literature on instrumental variables \citep{chernozhukov2015post,hansen2014instrumental}. The boundedness of both the largest and smallest eigenvalues guarantees that the design matrices $X$ and $Z$ are well-conditioned, thereby ensuring numerical stability and invertibility of relevant sample moment matrices. The restrictions on $\|\gamma_x\|$ and $\|\gamma_z\|$ are bounds of the magnitude of the coefficients helping to prevent overfitting and to maintain stability of the regression estimates. 
The eigenvalue restrictions are imposed in a probabilistic sense, holding with probability approaching one as $n \to \infty$, and are satisfied under sub-Gaussian designs provided that $p_x/n \to \tau_x < 1$ and $p_z/n \to \tau_z < 1$ \citep{vershynin2010introduction}. 

%\textcolor{orange}{(OLD DESCRIPTION) Assumption \ref{assumption_reg} requires that the signal retained after regularization comprises a non-negligible fraction of the signal present prior to regularization similar to the literature involving instruments \citep{chernozhukov2015post,hansen2014instrumental}. 
%In particular, this assumption requires that matrices $X$ and $Z$ are well-conditioned, and neither of the largest eigenvalues grows too large, nor the smallest eigenvalue approaches zero, ensuring numerical stability and invertibility.  The constraints on $\|\gamma_{x}\|$ and $\|\gamma_{z}\|$ ensure that the coefficients are bounded, which helps prevent overfitting and ensure the stability of the regression estimates. This assumption applies to the populations  $Z$  and  $X$; however, in the demonstrations, only the subsets  $Z^{(2)}$  and  $X^{(2)}$  are required to satisfy the condition, making this assumption sufficient.}

%\begin{assumption} \label{assumption_eta}
%The tuning parameter $\eta$ satisfies that $\eta \|\bgamma_x\|^2 =o(1)$
%%and  $\|n^{-1}\eta X^{\top}d\| =o(1)$.
%\end{assumption}

\begin{assumption} 
  If we write $Z=XB+W$,
    \label{assumption_sigma_E}
    \begin{enumerate}
        \item[(i)]  $\|B\| \leq C_B$, \ for some constant \ $C_B > 0$;
        \item[(ii)]  $\lambda_{\max}(\Sigma_{W}) < C_W$, $\lambda_{\min}(\Sigma_{W}) > c_W$, for some constants $C_W>0$ and $c_W > 0$.
        %The covariance matrix $\Sigma_{W}$ of the error term $W^{(2)}$ must be positive definite.
        %S\Kchange{Write it in the assumption2's style. Bounding the largest and smallest eigen value.}
       \end{enumerate}
\end{assumption}

%\XDchange{Assumption \ref{assumption_sigma_E} should be able to have enough assumption on $X$. I think my problem is, here I can say that $X$ is the ``pure'' control, then the $X$ in assumption \ref{assumption_reg} is the combination. The easiest way to fix is to write $Z_1=XB+W$ - I am not sure if that's the best way to do it}
%\XDchange{We need to make sure that $\|n_{2}^{-1}\eta_{x}B^{\top}B\|=o(1)$, thus we need $\|B\|=o(\sqrt{n/\eta_{x}})$, where $||$ is the spectral norm - It's weaker than the assumption above and the case with $I$ can satisfy it. I need to check other places in proof again.}
%\XDchange{Note: normalize the superscription of assumption}

%\textcolor{red}{In Assumption \ref{assumption_sigma_E}, $B$ captures the correlation between $Z$ and $X$. The boundedness on $B$ ensures that the projection of $Z$ onto the space spanned by $X$ is well-defined and does not result in excessively large projections. It also prevents the coefficients from becoming overly large, which could lead to numerical instability and unreliable estimates. The assumption on $W^{(2)}$ ensures that instruments $Z$ have enough information that cannot be explained by $X$.}

In Assumption \ref{assumption_sigma_E}, $B$ denotes the correlation structure between the instruments $Z$ and the covariates $X$. Bounding $\|B\|$ serves two purposes: it ensures that the projection of $Z$ onto the column space of $X$ remains well-behaved (avoiding arbitrarily large projections), avoiding unbounded projections, and excessively large coefficient estimates that may lead to numerical instability and obscure the model’s economic interpretation.
The conditions on the covariance matrix $\Sigma_W$ guarantee that the component of $Z$ orthogonal to $X$ contains sufficient independent variation. Specifically, bounding the largest and smallest eigenvalues guarantees that the residualized instruments neither exhibit excessive variance nor approach degeneracy. This well-conditioned residual component makes the instruments informative beyond what is explained by $X$, preserving identification strength and promoting the numerical stability of subsequent estimators.

%Assumption \ref{assumption_sigma_E} implies that the information contained within the correlation between the instrument and the endogenous error term is informative. 

%\begin{assumption} \label{assumption_eta2}
%The tuning parameter $\eta$ satisfies that $n^{1/2}\eta \|\gamma_x\|^2 =o(1)$
%%and  $\|n^{-1}\eta X^{\top}d\| =o(1)$.
%\end{assumption}

%\XDchange{Assumption \ref{assumption_eta2} imposes a rate of convergence on the tuning parameter $\eta$. LIT}

% \begin{assumption} \label{assumption_QQ}
%           $Q=Z^{\left(1\right)}(Z^{\left(1\right)\top}Z^{\left(1\right)}/n+\eta_{z}I)^{-1}Z^{\top}(I-A_{n})$. 
% \end{assumption}

%\textcolor{red}{MORE COMMENTS on Assumption \ref{assumption_QQ}.}
\begin{assumption}\label{assumption_large_p}
Assume that, as $(n,p_x,p_{z}) \to \infty$, $p_x/n \rightarrow \tau_x$ and $p_z/n = (p_{z_{1}}+p_x)/n  = (p_{z_{1}}/n+p_x/n)=\tau_{z_{1}}+\tau_x \rightarrow \tau_z$, where $\tau_{x}>0$ and $\tau_{z}>0$. 
\end{assumption}

Notice that Assumption \ref{assumption_large_p} allows for both the number of instruments and covariates to grow 
% faster than 
with the sample size. In particular, these ratios could even diverge to infinity. This assumption is an extension of that in \citet{hansen2014instrumental}. 

\begin{assumption} \label{assumption_homo}
$\E_{X}\left[\epsilon\epsilon^{\top}\right]=\sigma_{\epsilon}^{2} < \infty$
and $\E_{X}\left[vv^{\top}\right]=\sigma_{v}^{2}< \infty$, where $\E_{X}[\cdot]$ denotes the expectation conditional on $X$.
\end{assumption}

%Finally, Assumption \ref{assumption_homo} is a homoskedasticity condition.

Assumption \ref{assumption_homo} imposes a homoskedasticity condition. This constraint on the variance structure simplifies the model's error term properties, facilitating more straightforward and easily implementable inference procedures. We leave extensions to the heteroskedastic case to future research.

\subsection{Asymptotic results}

In this section, we establish the asymptotic properties of the general two-step estimator described in equation \eqref{eq: estimator} above. In particular, we establish the consistency and the asymptotic normality. We will discuss details on inference procedures in the next section.

%\Kchange{Use n not n2. n is the feature of dataset but n2 is derived from the proposed method.}
\begin{theorem}\label{thm:consistency}
Let Assumptions \ref{assumption_Eigen_Upsilon}--\ref{assumption_large_p} hold, and $\eta_{x}$ satisfies $n^{1/2}\eta_{x}\|\gamma_{x}\|=o(1)$. Moreover, let $\frac{\operatorname{tr}\{QQ^{\top}QQ^{\top}\}-\operatorname{tr}\{QQ^{\top}\circ QQ^{\top}\}}{\operatorname{tr}(QQ^{\top})-\operatorname{tr}(Q\circ Q)}=o_p(1)$, where \ensuremath{Q=Z(Z^{\top}Z/n+\eta_{z}I)^{-1}Z^{\top}(I-A_{n})}. Then, as $n\rightarrow\infty$,
\begin{equation*}
\widehat{\alpha} \overset{p}{\to}\alpha.
\end{equation*}
\end{theorem}

Theorem \ref{thm:consistency} shows the consistency of the proposed estimator, a desired property for most estimators. The variable $Q$, as presented in the theorem statement, represents a specific transformation of the instrumental variables matrix $Z$. This transformation involves regularization by $\eta_{z}$ followed by projection using $I-A_{n}$. We assume that the projection $Q$ behaves in a predictable manner, and the impact of complex interactions diminishes. This assumption is natural since
\begin{equation*}
QQ^{\top}=Z\left(Z^{\top}Z/n+\eta_{z}I\right)^{-1}Z^{\top}\left(I-A_{n}\right)^{2}Z\left(Z^{\top}Z/n+\eta_{z}I\right)^{-1}Z
\end{equation*}
%$QQ^{\top}=Z^{\left(1\right)}(Z^{\left(1\right)\top}Z^{\left(1\right)}/n+\eta_{z}I)^{-1}Z^{\top}(I-A_{n})^{2}Z(Z^{\left(1\right)\top}Z^{\left(1\right)}/n+\eta_{z}I)^{-1}Z^{\left(1\right)}$
is also a quadratic form. Note that this assumption pertains to the population $Z$, but in the proof below, only the subset  $Z^{(1)}$  is required. Therefore, this assumption is sufficient.

Here, the tuning parameter $\eta_x$ controls the extent of shrinkage applied to the coefficients rather than directly determining the sparsity of the model. In contrast to the $L_1$ norm requirements in  \cite{chernozhukov2015post}, where the tuning parameter is used to achieve varying levels of sparsity or to eliminate predictors, our method focuses on regulating the magnitude of the coefficients, allowing them to be small but non-zero. It is worth noting that in this context, we are imposing a stronger assumption of normality on the regression coefficients. By leveraging this sufficient assumption, we can effectively control the complexity of the model while maintaining the inclusion of all predictors, thereby striking a balance between model interpretability and predictive performance.

%\textcolor{red}{Those are theoretical assumptions for $\eta$, our empirical way of choosing it is following the literature and I will talk about it in theorem 3.}

%\begin{lemma}
%\label{lemma_A1n_randommatrix}
%Following the random matrix theory, 
%$$ \lim_{n \rightarrow \infty} n^{-1}\operatorname{tr}(A_{n}) = \tau \{1-\eta m(\eta)\} = O_p(1) $$
%\end{lemma}
%\begin{proof}
%See Lemma 2 of \citet{liu2020estimation}
%\end{proof} 

The following result establishes the asymptotic normality of the two-step estimator, which is crucial for conducting statistical inference. The limiting distribution of the estimator is shown to be a standard normal, which provides the foundation for the inference procedures discussed in the subsequent section. In order to construct confidence intervals and perform hypothesis tests, an easy-to-use consistent estimator for the asymptotic variance term, denoted by $\sigma_{\alpha}$, is required. The derivation and properties of such a variance estimator will be addressed in the next section.

\begin{theorem}\label{thm:asymp normality}
Let Assumptions \ref{assumption_Eigen_Upsilon}--\ref{assumption_homo} hold, $\eta_{x}$  satisfy $n_{2}^{1/2}\eta_{x}\|\gamma_{x}\|=o(1)$, and $\eta_{z}$ satisfy $n_{2}^{1/2}\eta_{z}\|\gamma_{z}\|=o(1)$.
%$\E_{X}\left(\epsilon\epsilon^{\top}\right)=\sigma_{\epsilon}^{2}$ and $\E_{X}\left(vv^{\top}\right)=\sigma_{v}^{2}$. 
Moreover, let $\frac{\operatorname{tr}\{QQ^{\top}QQ^{\top}\}-\operatorname{tr}\{QQ^{\top}\circ QQ^{\top}\}}{\operatorname{tr}(QQ^{\top})-\operatorname{tr}(Q\circ Q)}=o_p(1)$, where \ensuremath{Q=Z(Z^{\top}Z/n+\eta_{z}I)^{-1}Z^{\top}(I-A_{n})}. Then, as $n\to\infty$,
\begin{equation*}
\sqrt{n_{2}}\left(\widehat{\alpha}-\alpha\right)/\sigma_{\alpha} \overset{d}{\to}N\left(0,1\right),
\end{equation*}
where \
$\sigma_{\alpha}^{2}=\left(d^{\top}P_{z}^{\top}\left(I-A_{n}\right)d\right)^{-1}d^{\top}P_{z}^{\top}\left(I-A_{n}\right)P_{z}d\left(d^{\top}P_{z}^{\top}\left(I-A_{n}\right)d\right)\sigma_{\epsilon}^{2}$, and \ $P_{z}=Z^{\top}\left(Z^{\top}Z+\eta_{z}I\right)^{-1}Z$.
% $\sigma_{\alpha}^{2} = \left(d^{\top}\left(I-A_{n_{2}}\right)d\right)^{-1}d^{\top}\left(I-A_{n_{2}}\right)d (d^{\top}\left(I-A_{n_{2}}\right)d)^{-1}\sigma_{\epsilon}^{2}$. 
\end{theorem}

% --- \textcolor{orange}{(INCLUDE here an explanation about the $\sqrt{n}$-convergence relative to existing results: \citet{hansen2014instrumental} \citet{CSHNW12} and \citet{chernozhukov2018debiased}.)}

%\textcolor{red}{
%Theorem 2 establishes the asymptotic normality of the proposed estimator by integrating data-splitting with random-matrix theory, thereby validating its use for rigorous statistical inference. We discuss the estimation of the variance $\sigma_{\alpha}^{2}$ in the next section.} 
%\textcolor{red}{We note that unlike the theory developed in \citet{chernozhukov2018debiased}, our result allows both covariates and instruments to have non-sparse effects. Also, the above result establishes a $\sqrt{n}$-consistency of estimator. A similar result under a non-sparse instruments condition was previously shown by \citet{CSHNW12} under the ``many-instruments'' asymptotic framework where the strong identification of  the first
%stage signal is feasible. 
% Besides, they also study the weak instruments asymptotics where the convergence rate is slower than the $\sqrt{n}$ rate, which is out of the scope of this article.
%In addition, \citet{CSHNW12} and \citet{hansen2014instrumental} further proposed jackknife-type procedures for parameter inference. In contrast, the closed-form expressions for the proposed estimator  
%$\widehat \alpha$ and its variance  are considerably simpler and avoid the computational burden inherent in jackknife-based methods.
%}

\begin{remark}
\citet{CSHNW12} and \citet{hansen2014instrumental} both consider identification and estimation under strong and weak instrument scenarios. To illustrate their approach, consider the following stylized first-stage regression:
    \begin{align*}
       d_i & = \sqrt{r_n/n} Z_i + v_i,
    \end{align*}
where the convergence rate of the resulting estimators depends on the scaling factor $r_n$. When instruments are strongly relevant (i.e., $r_n = O(n)$), the coefficients do not shrink, and standard root-n convergence rates are achievable. In contrast, under weak instrument conditions where $r_n = o(n)$, the parameter of interest $\alpha$ becomes only more weakly identified.
For further discussion, we refer the reader to Assumption 2 in \citet{CSHNW12} for more details. 
In addition, \citet{hansen2014instrumental} extends the Jackknife Instrumental Variable Estimator (JIVE) to settings with a large number of instruments by adopting a ridge-regularized first-stage estimation.
%\citet{hansen2014instrumental} studies a ridge-type estimator in the first stage regression so that extend the JIVE estimator to the cases where the number of instruments exceeds the identical size.
Under the assumptions considered in this paper, both our proposed method and those of \citet{CSHNW12} and \citet{hansen2014instrumental} achieve root-n consistency. Notably, our estimator is computationally more tractable and exhibits greater robustness with respect to the choice of tuning parameters.
%Under the assumptions of this article, the proposed method and that of \citet{CSHNW12,hansen2014instrumental}  achieve sqrt-n convergence rate. Moreover, our proposed estimator is computational simpler and robust to the tuning parameter selection.
\end{remark}
% Theorem 2 establishes the asymptotic normality of the proposed estimator with data splitting and random matrix theory, 
% thereby confirming its validity for rigorous statistical inference. 
% Unlike the theory in \citet{chernozhukov2018debiased}, the result accommodates settings in which both covariates and instruments have non-sparse effects.
% We note that similar sqrt-n convergence rate for non-sparse iv is previous discussed by \citet{CSHNW12}, which is named as the many instrument asymptotic.
% \citet{CSHNW12,hansen2014instrumental} considers Jackniffe type estimators for inference.
% In comparison, the closed-form expressions for the proposed estimator $\widehat\alpha$ and its variance $\widehat\sigma$ are markedly simpler than the jackknife-based procedure of \citet{hansen2014instrumental} and avoid the attendant computational complexity.
% Theorem 2 provide the asymptotical normality, demonstrating the ability of the proposed method in conducting rigorous statistical inference.
% These results provide an
%  interesting comparison to the asymptotic results in \citet{hansen2014instrumental}  and \citet{chernozhukov2018debiased}.
%  Different from \citet{chernozhukov2018debiased}, Theorem 2 derive the result under the cases where the covariates and instruments are of non-sparse effects.
%  In addition,
% the proposed estimator $\widehat \alpha$ and its variance $\widehat \sigma$ are of simple form, which is much more straight than that of \citet{hansen2014instrumental}, which depends on jackknife and its variance estimator is computing complex.

We close this section emphasizing the importance of the sample splitting for the properties of the estimator. When considering a high-dimensional structure, the dependence between the first and second stages can lead to biased and inconsistent estimates. The data splitting ensures that the estimates from the first stage are independent of the second stage errors, thereby preserving the validity of the asymptotic properties. Next, we develop inference procedures.

\section{Inference}\label{sec:inference}

In this section, we turn our attention to inference procedures. Given the results on asymptotic normality of the estimator presented in Theorem \ref{thm:asymp normality}, inference for the treatment effects parameter is simple. The only remaining ingredient is a consistent estimator of the variance, $\sigma^2_{\alpha}$.  

We construct an estimator for the variance by employing the $RidgeVar$ technique for inference (\citet{liu2020estimation}). The proposed approach involves applying the ridge projection matrix to both the control variables (numerator) and the instrumental variables (denominator). In comparison to existing literature, our $RidgeVar$ method offers improved efficiency for conducting inference. 

From the asymptotic normality in Theorem \ref{thm:asymp normality}, we wish to estimate the parameter $\sigma^2_{\alpha}$. Now, using the sample-splitting together with the homoskedasticity condition in Assumption \ref{assumption_homo}, we have that $\E_{X}\left[\epsilon^{(2)}\epsilon^{(2)\top}\right]=\sigma_{\epsilon}^{2}$
and $\E_{X}\left[v^{\left(1\right)}v^{\left(1\right)\top}\right]=\sigma_{v}^{2}$. Thus, the following sandwich-type variance estimator can be derived,
\begin{align*}
\widehat{\sigma}_{\alpha,(2)}^{2} & =\textnormal{Var}\left(\left(\widehat{d}^{(2)\top}\left(I-A_{n_{2}}\right)d^{(2)}\right)^{-1}\widehat{d}^{(2)\top}\left(I-A_{n_{2}}\right)\epsilon^{(2)}\right)\\
 & =\left(\widehat{d}^{(2)\top}\left(I-A_{n_{2}}\right)d^{(2)}\right)^{-1}\widehat{d}^{(2)\top}\left(I-A_{n_{2}}\right)^{2}\widehat{d}^{(2)}\left(\widehat{d}^{(2)\top}\left(I-A_{n_{2}}\right)d^{(2)}\right)^{-1}\widehat{\sigma}_{\epsilon,(2)}^{2}\\
 & \approx\left(\widehat{d}^{(2)\top}\left(I-A_{n_{2}}\right)d^{(2)}\right)^{-1}\widehat{d}^{(2)\top}\left(I-A_{n_{2}}\right)\widehat{d}^{(2)}\left(\widehat{d}^{(2)\top}\left(I-A_{n_{2}}\right)d^{(2)}\right)^{-1}\widehat{\sigma}_{\epsilon,(2)}^{2}.
\end{align*}
The approximation in the last line of the above display arises from the assumption that $(I-A_{n_{2}})^2 \approx (I-A_{n_{2}})$. This simplification is justified when $A_{n_{2}}$ is a projection matrix or when its eigenvalues are close to 0 or 1, leading to negligible differences between its square and itself.

Therefore, we propose to use the following estimator: 
\begin{equation} \label{eq variance estimator}
\widehat{\sigma}_{\alpha}^{2}=\left(\widehat{d}^{(2)\top}\left(I-A_{n_{2}}\right)d^{(2)}\right)^{-1}\widehat{d}^{(2)\top}\left(I-A_{n_{2}}\right)\widehat{d}^{(2)}\left(\widehat{d}^{(2)\top}\left(I-A_{n_{2}}\right)d^{(2)}\right)^{-1}\widehat{\sigma}_{\epsilon,(2)}^{2},
\end{equation}
where the estimator for the variance of the error term is given by 
\begin{equation}\label{eq variance error}
\widehat{\sigma}_{\epsilon,(2)}^{2}=\frac{y^{(2)}\left(I-P_{n_{2}}\right)y^{(2)}/n_{2}}{1-\operatorname{tr}\left(P_{n_{2}}\right)/n_{2}},    
\end{equation}
 with $P_{n_{2}}=S_{n_{2}}\left(S_{n_{2}}^{\top}S_{n_{2}}+\eta_{s}I\right)^{-1}S_{n_{2}}^{\top}$,
and $S_{n_{2}}=\left[d^{(2)},X^{(2)}\right]$. 

%\textcolor{red}{(CHECK ON THE $\widehat{d}$ used for estimation of the variance. I don't think we need $\widehat{d}$ in $\widehat{\sigma}_{\epsilon,(2)}^{2}$. The reason is that it already uses the correct IV coefficients.)}

The practical estimation of the variance of the error term in equation \eqref{eq variance error} is simple, and consequently, the estimation of the variance in equation \eqref{eq variance estimator} is straightforward. %It can be estimated by $\widehat{\sigma}_{\epsilon,(2)}^{2} = \left\{1 - n_{2}^{-1} \text{tr}(P_{2n})\right\}^{-1} \check{\sigma}_\epsilon^2$, where $\check{\sigma}_\epsilon^2=y^{(2)}\left(I-P_{n_{2}}\right)y^{(2)}/n_{2}$. 
This method relies on the methods developed in \citet{liu2020estimation}. 
They suggest a consistent estimator of variance based on ridge regression and random matrix theory, which is valid under both low- and high-dimensional models. We extend their results to cases where both the regressors and instruments are of high dimension.

The following theorem formally establishes the consistency of the variance estimator $\widehat{\sigma}_{\alpha}^{2}$ in \eqref{eq variance estimator}.

\begin{theorem}\label{thm:variance} 
Under Assumptions \ref{assumption_Eigen_Upsilon}--\ref{assumption_homo}, as $n \to \infty$, we have that
\begin{align*}
\widehat{\sigma}^2_{\alpha} & \overset{p}{\to} \sigma^2_{\alpha}.
\end{align*}
\end{theorem}

%\begin{proof}
%    See Appendix \ref{proofs_theorem3}
%\end{proof}

Given such an estimator, it is possible to formulate a wide variety of tests and construct confidence intervals.
General hypotheses on the vector $\alpha$ can be accommodated by standard Wald-type tests. These statistics and their associated limiting theory provide a natural foundation for testing the null hypothesis $H_{0}: \alpha = r$ when $r$ is known. Thus, the result in Theorem \ref{thm:asymp normality} can be used to test the null hypothesis $H_{0}$.

Thus, the Wald statistic can be constructed as
\begin{equation}
\label{eq:inf2}
\mathcal{W} = n \frac{(\widehat{\alpha}-r)^2}{ \widehat{\sigma}_{\alpha}^{2}},
\end{equation} 
where $\widehat{\sigma}_{\alpha}^{2}$ is given in equation \eqref{eq variance estimator}.

\begin{corollary} \label{theorem:Wald} (Wald Test Inference). Under Assumptions \ref{assumption_Eigen_Upsilon}--\ref{assumption_homo}, and $H_{0}: \alpha=r$, 
\begin{equation*}
\mathcal{W} \stackrel{a}{\sim} \chi_{1}^{2}.
\end{equation*}
\end{corollary}

\begin{remark}
    We make use of sample splitting for estimation, as discussed previously. We note that, naturally, multiple sample splits generate many corresponding p-values. Hence, one inference strategy would be to combine these p-values to derive one unique inference result 
    (see, e.g., \citet{liu2020cauchy}, and \citet{vovk2022admissible}).
    %Xihong Lin Harvard. very famous and easy to use
    % (INCLUDE A REFERENCE). 
    Another popular recent strategy would be to employ a Cauchy combination method to produce inference (see, e.g., \citet{liu2020cauchy}). We leave these extensions for future research.
\end{remark}

\section{Simulation}\label{sec:simulation}

This section presents a Monte Carlo simulation study designed to assess the finite sample performance of the proposed estimation and inference methods. Overall, this simulation study aims to provide insights into the small sample behavior of the new procedures under different data generating process (DGP) specifications, taking into account the endogeneity of the variable of interest, and the presence of high dimension for both instrumental variables and exogenous regressors.

We first introduce the DGP of variables used in the simulations.
The outcome variable $y_i$ is generated according to the following model:
\begin{equation*}
y_i = \alpha d_i + \gamma_x^{\top}X_i + \epsilon_i,
\end{equation*}
where $d_{i}$ is the endogenous variable of interest, the vector $X_i$ contains the completely exogenous regressors, and $\epsilon_i$ is the error term. In all simulations, the true value of the parameter of interest, $\alpha$, is set to be 1.
The endogeneity of the variable $d_i$ is generated according to the equation:
\begin{equation*}
d_i = \gamma_z^{\top}Z_i + v_i,
\end{equation*}
where $Z_i$ is a vector of instrumental variables,
%, which includes both the exogenous variables $X_i$, and an additional instrument $Z_{1i}$ 
and $v_i$ is the error term. 
We will discuss the specific details about the choice of the parameter values for $\gamma_x$ and $\gamma_z$ below.

The error terms $\epsilon_i$ and $v_i$ are jointly normally distributed with a mean vector of zeros and a covariance matrix given by:
\begin{equation*}
\left(\epsilon_{i},v_{i}\right)\sim N\left(\left(\begin{array}{c}
0\\
0
\end{array}\right),\left(\begin{array}{cc}
1 & \sigma\\
\sigma & 1
\end{array}\right)\right).
\end{equation*}
This covariance structure allows for correlation between the error terms, with a correlation coefficient $\sigma = 0.6$.

This design framework allows us to investigate the relative performance of various estimators under different sparsity patterns while controlling for the overall signal strength. The combination of these structures with correlation patterns (EC and AR(1)), described below, provides a comprehensive setting for examining the behavior of high-dimensional instrumental variable estimators under varying degrees of sparsity and correlation.

In the simulations that follow, we introduce two distinct DGPs, one designed for non-sparse settings, and another for sparse. We evaluate the proposed two-step ridge regression (\textbf{TSRR}) estimator for the parameter of interest discussed in equation \eqref{eq: estimator}, as well as its corresponding variance estimator in \eqref{eq variance estimator}. For comparison, we also include results for high-dimensional methods (\textbf{DML}) presented in \citet{chernozhukov2015post} and \citet{belloni2017program}; 
and the regularized jackknife instrumental variable estimator (\textbf{RJIVE}) suggested in \cite{hansen2014instrumental}.\footnote{The DML is implemented with the R package \textit{hdi}, and the RJIVE is implemented by ourselves. All the codes are available upon request.}
We present results for the following empirical statistics: bias (for point estimates), bias of the variance, mean squared error (MSE), the probability of the true parameter value being inside the confidence interval (P(cover)) for a 95\% nominal coverage, the length of the confidence interval, and the computation time. Finally, in the simulations, the number of replications is set to $500$. 

% The proposed TSRR estimation requires the use of nuisance parameters. 
As discussed previously, $\eta = (\eta_x,\eta_z)$,
the penalty parameter of TSRR,
is chosen adaptively as the minimum of two data-driven values: one based on the correlation between controls and the outcome, i.e., %$\max(X^{\left(2\right)\top}y^{\left(2\right)})$
$\eta_{x}=c_{x} \max_{1 \leq j \leq p} \big| x_{(j)}^{\top} Y \big| / (n_{2} p_{x})$; and another based on the correlation between controls and the endogenous variable, i.e., %$\max(X^{\left(2\right)\top}D^{\left(2\right)})$. 
$\eta_{z}=c_{z} \max_{1 \leq j \leq p} \big| z_{(j)}^{\top} d \big| / (n_{1} p_{z})$, and , where $c_{z}$ and $c_{x}$ are preset constants.
Both values are scaled by the sample size, dimension, and a tuning parameter $c_{\cdot} \in \{0.1,1\}$. This dual-equation approach ensures appropriate penalization for both the structural and first-stage relationships.

We evaluate the finite-sample performance of our proposed estimator through Monte Carlo simulations. Unless otherwise noted, the sample size is set to $n=500$, with $p_{x}=700$ control variables and $p_{z}=500$ instruments. All simulations are repeated 1000 times. We examine several data-generating processes (DGPs) that vary in both the strength  and structure of the instruments as well as the correlation patterns among covariates. For each design, we report the average bias, root mean squared error (RMSE), and empirical coverage of the nominal 95\% confidence intervals.

To examine the robustness of our estimator, we consider several distinct data-generating processes (DGPs) that differ in instrument strength and dependence structure, including an ``all-weak'' design, a ``cutoff'' structure, and a ``sparse'' case with an autoregressive AR(1) correlation pattern.  The detailed specifications of these scenarios are given below.

\subsection{Non-sparse case}
The first DGP, which we label the ``non-sparse'' case, features an equi-correlation (EC) structure, where all variables exhibit a common, non-negligible level of pairwise correlation. Specifically, the correlation matrix, $\Sigma$, follows an equi-correlation structure given by $\Sigma_{i\neq j}^{EC}=\rho$ and $\Sigma_{i,i}^{EC}=1$. This structure ensures that all variables maintain a consistent correlation level $\rho$. We examine the case with equi-correlation structure $\rho = 0.04$.

The coefficient vector for the control variables is: $\gamma_x = (m,\ldots,m,c\xi_1,\ldots,c\xi_s,0,\ldots,0)^{\top}$, with 5 strong signal components of magnitude $m$, and 5\% of the remaining coefficients are non-zero. These non-zero coefficients $\xi_j$ are independently drawn from $N(0,1)$ and scaled by constant $c$, which is calibrated to achieve a specified signal strength $\mu_x^2 \in \{300,600\}$. The scaling constant $c$ is chosen such that:
\begin{equation*}
c = \sqrt{\frac{\mu_x^2}{n\xi^{\top}\Sigma_x\xi + \mu_x^2\xi^{\top}\Sigma_x\xi}},
\end{equation*}
where $\mu_x^2$ represents the signal strength following \cite{chernozhukov2015post}, and $\Sigma_x$ denotes the correlation matrix for control variables.

Moreover, in this setting, we have a moderately high-dimensional control variable space with $p_{x}=700$ covariates and a higher-dimensional instrument space of $p_{z}=500$ variables. The instrument coefficients maintain an ``all-weak" structure, where $\gamma_z = (c\xi_1,\ldots,c\xi_s,0,\ldots,0)^{\top}$ with 50\% of the coefficients being non-zero and scaled to achieve a signal strength of $\mu_z^2 = 600$, which measures the overall strength of the instruments. In another ``cut-off" structure, the instrumental variables' coefficients follow a cutoff structure where exactly 70\% of the instruments are non-zero, while the remaining 30\% are zero coefficients. Specifically, we set the first 70\% of the instruments to have equal coefficients (initially set to 1) and the remaining instruments to have zero coefficients. These coefficients are then rescaled by a constant to achieve the same concentration parameter ($\mu_z^2 = 600$) as the ``all-weak" setting. It imposes a deterministic pattern where the relevance of instruments follows a strict ordering.

\begin{table}[ht]
\centering
\begin{tabular}{lcccccc}
  %\multicolumn{7}{l}{\textbf{Panel A: allweak(0.5)}} \\
  %\hline
  % & Bias & Bias (Var.) & MSE & P-Cover & Length & Time \\ 
  %\hline
  %TSRR & -0.00797 & 0.02298 & 0.10546 & 0.94200 & 1.18254 & 0.07230 \\ 
 % RJIVE & -0.02225 & 0.09261 & 0.10041 & 0.94200 & 0.87601 & 1.89933 \\ 
 % cherno & 0.16777 & 0.28526 & 0.40158 & 0.89800 & 1.27725 & 0.09952 \\ 
  %\hline
  %\\
  \multicolumn{7}{l}{\textbf{Panel A: Strong Instruments with Moderate Control Sparsity}} \\
  %HD-constant-0.04-Control-300-0.3-5-2-IV-cutoff-0.7-mu2-600-eta0-0.1-alpha-1-looppart-1
  \hline
   & Bias & Bias (Var.) & MSE & P-Cover & Length & Time \\ 
  \hline
  TSRR & -0.0184 & 0.0252 & 0.1249 & 0.9540 & 1.2844 & 0.2367 \\ 
  RJIVE & 0.6167 & -0.0066 & 0.9698 & 0.8540 & 3.0355 & 2.1053 \\ 
  DML & -0.0542 & 0.0109 & 0.0109 & 0.8780 & 0.3068 & 0.0919 \\ 
  \hline
 \\
  \multicolumn{7}{l}{\textbf{Panel B: Strong Instruments with High Control Sparsity}} \\
  %HD-constant-0.04-Control-300-0.05-5-2-IV-allweak-0.3-mu2-600-eta0-0.1-alpha-1-looppart-1
  \hline
   & Bias & Bias (Var.) & MSE & P-Cover & Length & Time \\ 
  \hline
  TSRR & 0.0077 & 0.0038 & 0.0958 & 0.9585 & 1.1980 & 0.4354 \\ 
  RJIVE & -0.0442 & 0.2455 & 1.0903 & 0.8685 & 3.1271 & 3.6084 \\ 
  DML & 0.9060 & 0.5472 & 21.4828 & 0.8547 & 15.6733 & 0.3636 \\ 
  \hline
  \\
   \multicolumn{7}{l}{\textbf{Panel C: Standard Regularization Strength}} \\
 % HD-constant-0.04-Control-300-0.05-5-2-IV-allweak-0.3-mu2-600
  \hline
   & Bias & Bias (Var.) & MSE & P-Cover & Length & Time \\ 
  \hline
  TSRR & -0.0173 & 0.0745 & 0.1893 & 0.9660 & 1.4124 & 0.1167 \\ 
  RJIVE & 0.0403 & 0.1857 & 0.9615 & 0.8620 & 3.1126 & 1.1469 \\ 
  DML & 0.6845 & 0.4838 & 1.6292 & 0.8520 & 2.3269 & 0.0588 \\ 
  \hline
\end{tabular}
\caption{The non-sparse case}
\label{tab:combined nonsparse}
\end{table}
%Simulation Results for Different Settings. Panel A presents results for allweak(0.5) setting with N=400, Pcontrol=700, PIV=500, Control=300. Panel B presents results for the cutoff(0.7) setting with the same parameters. Bias is calculated as rec\_est - 1, Bias (Var.) is taken from var\_sd, Length is from rec\_CI, and P-Cover is from rec\_cover.

The results for the first DGP, the less sparse controls, are collected in Table \ref{tab:combined nonsparse}. The upper part of Table \ref{tab:combined nonsparse}, Panel A, presents results for the case where the instruments follow a cutoff structure with 70\% active instruments, representing a scenario with strong identification. The middle panel, Panel B, examines a setting with a constant correlation structure among covariates and instruments that follow an all-weak pattern where 30\% of instruments are active but with smaller magnitudes. The lower panel, Panel C, investigates the case where $c_{x}=c_{z}= 1$, where $c_{\cdot}$ is a constant tuning parameter that determines the relative magnitude of regularization $\eta$. With $c_{\cdot}= 1$, the penalty parameters are directly proportional to the maximum correlations in the data, scaled by sample size and dimension. This serves as a reference point for evaluating how different choices of tuning parameters affect the regularization strength.

In this non-sparse case, our proposed estimator demonstrates superior performance compared to both alternative methods in all cases. Our approach yields very small empirical biases for both cases. For Panel A the point estimate bias is $-0.0184$, while the bias of variance is $0.0252$. In Panel B the bias for point estimate and variance are $0.0077$ and $0.0038$, respectively. In Panel C bias of point estimate and variance for TSRR are the smallest among all estimators. In addition, the TSRR estimator shows minimal bias and maintains excellent empirical coverage properties, about 95.4\%, 95.8\%, and 96.6\%, for Panels A, B, and C, respectively. In contrast, the RJIVE estimator produces biased results. In Panel A the point estimates of RJIVE are severely biased, while in Panels B and C the point estimates have relatively smaller bias, but large bias in the variance. This lack of precision reflects on the empirical coverage of the confidence intervals, which are only about 86\% for the nominal 95\%. Finally, the 
DML estimator shows deteriorating performance with higher bias, and its coverage probabilities (88\%, 86\%, and 85\%) fall notably below the nominal level.

\subsection{Sparse case}

The second DGP, the ``sparse" case, maintains the same dimensionality as the first setting, with $p_{x}=700$ control variables and $p_{z}=500$ instruments. However, the correlation structure now follows an autoregressive AR(1) pattern, where $\Sigma_{i,j}^{AR(1)}=\rho^{|i-j|}$. This implies a sparser dependence structure, where variables are only correlated with their immediate neighbors, and the correlation strength decays exponentially with distance. The control variable coefficients in this case follow below: 
$\gamma_x = (c\xi_1,\ldots,c\xi_s,0,\ldots,0)^{\top}$, with 30\% of the remaining coefficients are non-zero but relatively weaker in magnitude. The instrument coefficient structure remains consistent with the first DGP, maintaining the ``cutoff" specification with the same signal strength calibration. Specifically, in an AR(1) structure, each variable is only correlated with its immediate neighbors, and the strength of the correlation decreases exponentially as the distance between the variables increases.

\begin{table}[ht]
\centering
\begin{tabular}{lcccccc}
 \\
  \multicolumn{7}{l}{\textbf{Panel A: Moderate Control Sparsity}} \\
  %AR1-0.5-Control-300-0.3-5-2-IV-cutoff-0.7-mu2-600
  \hline
   & Bias & Bias (Var.) & MSE & P-Cover & Length & Time \\ 
  \hline
  TSRR & -0.0093 & -0.0007 & 0.1579 & 0.9680 & 1.5601 & 0.1907 \\ 
  RJIVE & -0.0309 & 0.0538 & 1.1750 & 0.9680 & 4.0366 & 1.3246 \\ 
  DML & -0.0148 & 0.0025 & 0.0091 & 0.9460 & 0.3590 & 0.0825 \\ 
  \hline
  \\
  \multicolumn{7}{l}{\textbf{Panel B: High Control Sparsity}} \\
  %vAR1-0.5-Control-300-0.05-5-2-IV-cutoff-0.7-mu2-600
  \hline
   & Bias & Bias (Var.) & MSE & P-Cover & Length & Time \\ 
  \hline
  TSRR & -0.0640 & 0.0306 & 0.2375 & 0.9480 & 1.7739 & 0.2189 \\ 
  RJIVE & -0.1020 & 0.0045 & 1.0992 & 0.9680 & 4.0728 & 1.5158 \\ 
  DML & -0.0220 & 0.0002 & 0.0085 & 0.9480 & 0.3500 & 0.0957 \\ 
  \hline
  \\
    \multicolumn{7}{l}{\textbf{Panel C: Standard Regularization Strength}} \\
  %AR1-0.5-Control-300-0.3-5-2-IV-cutoff-0.7-mu2-600
  \hline
   & Bias & Bias (Var.) & MSE & P-Cover & Length & Time \\ 
  \hline
  TSRR & -0.0208 & 0.0300 & 0.1950 & 0.9660 & 1.6115 & 0.1014 \\ 
  RJIVE & 0.0353 & 0.0383 & 1.1539 & 0.9680 & 4.0585 & 0.5304 \\ 
  DML & -0.0187 & 0.0044 & 0.0095 & 0.9420 & 0.3566 & 0.0299 \\ 
  \hline
\end{tabular}
\caption{The sparse case} %$\eta_{\alpha}=0.1$.
\label{tab:combined sparse}
\end{table}
%Simulation Results for Different Settings with AR1(0.5). Panel A presents results for allweak(0.3) setting with N=400, Pcontrol=700, PIV=500, Control=300. Panel B presents results for cutoff(0.7) setting with the same parameters. Bias is calculated as rec\_est - 1, Bias (Var.) is taken from var\_sd, Length is from rec\_CI, and P-Cover is from rec\_cover.

Table \ref{tab:combined sparse} collects the results for the second specification for the case with sparse (AR(1)) correlation structure ($\rho = 0.5$) and sparse control variables, where the first five coefficients have strong signals (magnitude of 2). There are three configuration in this table. Panel A examines the case where the control variables exhibit moderate sparsity, with 30\% of the control coefficients being non-zero, alongside 5 strong signals of magnitude 2. Panel B presents a more sparse setting where only 5\% of the control coefficients are non-zero, while maintaining the same 5 strong signals of magnitude 2. In both panels, the instruments follow a cutoff structure where the first 70\% of instruments have non-zero coefficients. This design allows us to examine how the sparsity of control variables affects estimation while maintaining consistent instrument strength and correlation patterns. Panel C investigates the case where $c_{x}=c_{z} = 1$, where the penalization term in the objective function is neither amplified nor dampened, providing a natural benchmark for assessing the impact of regularization strength on the estimator's performance.
In this setting, our proposed estimator demonstrates comparable performance to DML estimator, with both methods achieving empirical coverages close to the desired 95\% nominal coverage probability. While DML shows lower MSE for all panels, both methods effectively estimate the true parameter value of 1, with TSRR presenting relatively small biases. Notably, both methods substantially outperform RJIVE, which exhibits considerably higher MSE for all three panels, and wider length for confidence intervals.

These results highlight two key strengths of our proposed estimator. First, it maintains competitive performance with state-of-the-art methods, such as DML in sparse settings. Second, and more importantly, it demonstrates superior robustness across different data generating processes, maintaining high accuracy and appropriate coverage even in non-sparse settings where alternative methods may struggle. This robustness to both sparse and non-sparse scenarios, combined with its computational efficiency, makes our method particularly attractive for empirical applications where the true data structure is unknown.

\section{Empirical application}\label{sec:application}

This section provides an empirical application to illustrate the proposed two-step ridge regression (TSRR) methods in this paper.
Following \cite{angrist1991does} and \citet{hansen2014instrumental}, we revisit the classic example in the many-instrument literature, which also contains high-dimension covariates. This example focuses on estimating the causal effect of schooling on earnings by addressing the potential endogeneity of schooling through the use of quarters of birth and interactions as instrumental variables. While maintaining a similar methodological framework, our analysis differs in several key aspects from the original study. We use a subset of the \cite{angrist1991does} dataset, focusing on individuals under 40 years of age, and employ a sample-splitting approach, where we use 5\% of the total sample. This dataset exploits an unusual natural experiment to estimate the impact of compulsory schooling laws in the United States. The experiment stems from the fact that children born in different months of the year start school at different ages, while compulsory schooling laws generally require students to remain in school until their sixteenth or seventeenth birthday. In effect, the interaction of school entry requirements and compulsory schooling laws compel students born in certain months to attend school longer than students born in other months. Because one's birthday is unlikely to be correlated with personal attributes other than age at school entry, the quarter of birth generates exogenous variation in education that can be used to estimate the impact of compulsory schooling on education and earnings. We refer readers to \cite{angrist1991does} for further details on the experiment and data.

Our instrumental variables framework follows the original linear specification for the structural equation:
\begin{equation*}
\log(w_i) = \beta_0 + \beta_1 s_i + X_i'\gamma + \epsilon_i,
\end{equation*}
with a standard linear first-stage regression
\begin{equation*}
s_i = \pi_0 + Z_i'\pi_1 + X_i'\pi_2 + v_i,
\end{equation*}
where $w_i$ is the weekly wage, $s_i$ represents years of education, $X_i$ is a vector of control variables, and $Z_i$ is the vector of instruments. Moreover, the error terms satisfy:
\begin{equation*}
\E[\epsilon_i, v_i|X_i,Z_i] = 0.
\end{equation*}

The control variables, $X_i$, include year-of-birth dummies, state-of-birth dummies, and their interactions, resulting in a dimension of $p_X = 510$. Our set of instruments, $Z_i$, comprises three quarter-of-birth dummies and their interactions with the control variables, yielding a total of $p_Z=1527$ instruments. We present results for a smaller configuration of the number of instruments as well. In addition, for comparison, we present results for the  RJIVE estimator.\footnote{The RJIVE estimator for $\beta_1$ takes the form: $\widehat{\beta}_{RJIVE} = (D'P_{Z\lambda}D)^{-1}D'P_{Z\lambda}y$,
where $P_{Z\lambda}$ is the regularized projection matrix with ridge penalty $\lambda$.}

% \begin{table}[ht]
% \centering
% \begin{tabular}{cccc}
% \hline
% & \# of IV & Estimates & Confidence Interval         \\
%  \hline
% \multirow{3}{*}{TSRR} & 3    &    0.29 &  [-0.22, 1.01] \\
%  & 180  & 0.17  & [-0.65, 1.01] \\
%  & 1527     & 0.18 & [0.01,0.36] \\ \hline
%  \multirow{3}{*}{RJIVE} & 3      &  0.11 &   [0.07,0.15]\\
%  & 180      & 0.11 & [0.08,0.14]\\
%  & 1527     & 0.11 & [0.07,0.14]\\
%  \hline
% \end{tabular}
% \caption{Effects of schooling on earnings}
% \label{tab:applic}
% \end{table}
\begin{table}[ht]
\centering
\begin{tabular}{cccc}
\hline
& \# of IV & Estimates & Confidence Interval         \\
 \hline
\multirow{2}{*}{TSRR} 
 & 180  & 0.17  & [-0.65, 1.01] \\
 & 1527     & 0.18 & [0.01,0.36] \\ \hline
 \multirow{2}{*}{RJIVE}
 & 180      & 0.11 & [0.08,0.14]\\
 & 1527     & 0.11 & [0.07,0.14]\\
 \hline
\end{tabular}
\caption{Effects of schooling on earnings}
\label{tab:applic}
\end{table}
The estimation results are collected in Table \ref{tab:applic}, where we present the proposed two-step ridge regression (TSRR), and for completeness we present the RJIVE results from \citet{hansen2014instrumental}. The main results for the TSRR model show that our estimates are positive across all specifications, with the point estimates about $0.17$. 
Most notably, when using the full set of instruments ($1527$ IVs), we obtain a statistically significant estimate of $0.18$ with a 95\% confidence interval of $[0.01, 0.36]$. 
Economically, this result indicates a positive return to schooling, with an additional year of education associated with approximately 18\% higher weekly wages. The estimates from the smaller instrument set (with 180 IVs) yield similar point estimates but with wider confidence intervals that include zero. Statistically, this result is very intuitive, since the proposed procedures are designed to accommodate a large number of instruments and exogenous covariates.
In comparison, the RJIVE estimates show invariant results across different instrument specifications. 
% Using 3 instruments (quarter-of-birth dummies), their RJIVE estimate was 0.11 (SE: 0.02). 
With 180 instruments (including quarter-of-birth interactions with year-of-birth and state-of-birth main effects), they obtained an RJIVE estimate of 0.11 (SE: 0.02). In their most comprehensive specification with 1,527 instruments (including all possible interactions), the RJIVE estimate was 0.1067 (SE: 0.02).

Overall, our TSRR approach yields results that are somewhat larger in magnitude than those from the RJIVE method in the original paper, particularly when using the full set of instruments. Our estimate of 0.18 is approximately 69\% higher than the RJIVE estimate of 0.1067, though there is some overlap in the confidence intervals. The statistical significance of our full-instrument model suggests that the TSRR method can effectively handle high-dimensional settings while producing economically meaningful estimates that are broadly consistent with the literature on returns to schooling.

\section{Conclusion}\label{sec:conlusion}

This paper proposes novel ridge regularization-based methods for estimating treatment effects and conducting inference in the presence of both high-dimensional instrumental variables and high-dimensional control variables. An advantage of these methods is that they are valid with or without sparsity assumptions. To mitigate high-dimensional instrument biases, we suggest a data splitting strategy for constructing an optimal instrument. The data splitting ensures that the estimates from the first stage are independent of the second stage errors, thereby preserving the validity of the asymptotic properties of the estimator.
We establish the statistical properties of the estimator, namely, consistency and asymptotic normality. We also propose an estimator for the variance and establish its consistency.

Moreover, we study the finite sample performance of the proposed methods using numerical simulations. Results document evidence that the proposed estimator has good finite sample properties, outperforming existing sparsity-based approaches across a variety of settings. The overall numerical results provide evidence of the resilience and adaptability of our method to changes in the sparsity structure, highlighting its robustness compared to existing approaches. Finally, we provide an empirical application to estimating the causal effect of schooling on earnings by addressing potential endogeneity through the use of high-dimension instrumental variables and high-dimension covariates.

\newpage

\appendix
\section{Appendix}

This appendix collects the proofs of the main results in the text.

\subsection{Proof of Theorem \ref{thm:consistency}}

First we present the proof of Theorem \ref{thm:consistency}, which states the consistency result of the proposed estimator.

%\XDchange{I first make the notation in this section to be consistent, and I will update the notation in the model part later}
\begin{proof}
In the first-stage estimation, we have $\widehat{\gamma}_z$ which is estimated with the subsample $\mS_1$ with size $n_1$.
Specifically,
\begin{align*}
\widehat{\gamma}_{z}=n_{1}^{-1}(n_{1}^{-1}Z^{\left(1\right)\top}Z^{\left(1\right)}+\eta_{z}I)^{-1}Z^{\left(1\right)\top}d^{\left(1\right)}.
\end{align*}
 % is estimated from the dataset $\mS_{1}$. 
% From the first-stage of regression,
% Then with the estimated $\widehat \gamma_z$,
We predict the estimated optimal instruments $\widehat{d}$ for the dataset $\mS_{2}$
with the estimated $\widehat\gamma_z$.
% using
% the parameter estimated from dataset $\mS_{1}$,
\begin{equation*}
\widehat{d}^{(2)}=Z^{(2)}\widehat{\gamma}_{z}.
\end{equation*}
Thus, we have the following
\begin{align*}
\widehat{\alpha} 
& =
\frac{
\widehat{d}^{(2)\top}(I-A_{n_{2}})d^{(2)}
}{
\widehat{d}^{(2)\top}(I-A_{n_{2}})Y^{(2)}}\\
& =
\frac{
\widehat{d}^{(2)\top}(I-A_{n_{2}})(d^{(2)}\alpha + X^{(2)}\gamma_x + \epsilon^{(2)})
}{
\widehat{d}^{(2)\top}(I-A_{n_{2}})d^{(2)}
}\\
& =
\alpha + 
\frac{
\widehat{d}^{(2)\top}(I-A_{n_{2}})(X^{(2)}\gamma_x + \epsilon^{(2)})
}{
\widehat{d}^{(2)\top}(I-A_{n_{2}})d^{(2)}
}\\
% & =
% \left(
% n_{2}^{-1}
% \{1-n_{2}^{-1}\tr(A_{n_{2}})\}^{-1}
% \widehat{d}^{(2)\top}(I-A_{n_{2}})d^{(2)}
% \right)^{-1}
% n_{2}^{-1}
% \{1-n_{2}^{-1}\tr(A_{n_{2}})\}^{-1}
% \widehat{d}^{(2)\top}(I-A_{n_{2}})Y^{(2)}\\
 & =\alpha+\left(\frac{n_{2}^{-1}}{1-n_{2}^{-1}\operatorname{tr}\left(A_{n_{2}}\right)}\widehat{d}^{(2)\top}(I-A_{n_{2}})d^{(2)}\right)^{-1}\frac{1}{1-n_{2}^{-1}\operatorname{tr}\left(A_{n_{2}}\right)}\left[\underset{T_{1}}{\underbrace{n_{2}^{-1}\widehat{d}^{(2)\top}(I-A_{n_{2}})X^{(2)}\gamma_{x}}}\right.\\
 & \ +\left.\underset{T_{2}}{\underbrace{n_{2}^{-1}\widehat{d}^{(2)\top}(I-A_{n_{2}})\epsilon^{2}}}\right]\\
 & =\alpha+\left(\frac{n_{2}^{-1}}{1-n_{2}^{-1}\operatorname{tr}\left(A_{n_{2}}\right)}\widehat{d}^{(2)\top}(I-A_{n_{2}})d^{(2)}\right)^{-1}\frac{1}{1-n_{2}^{-1}\operatorname{tr}\left(A_{n_{2}}\right)}\left(T_{1}+T_{2}\right),
 % \\
 % & =\alpha+o_{p}\left(1\right),
\end{align*}
where $A_{n_{2}}=X^{(2)}\left(X^{(2)}{}^{\top}X^{(2)}+\eta_{x}I\right)^{-1}X^{(2)\top}$. 
For the consistency property, it suffices to establish the following  result
\begin{equation}
\left(\frac{n_{2}^{-1}}{1-n_{2}^{-1}\operatorname{tr}\left(A_{n_{2}}\right)}\widehat{d}^{(2)\top}(I-A_{n_{2}})d^{(2)}\right)^{-1}\frac{1}{1-n_{2}^{-1}\operatorname{tr}\left(A_{n_{2}}\right)}\left(T_{1}+T_{2}\right)=
o_{p}\left(1\right).
\end{equation}
% given that $\frac{1}{1-n_{2}^{-1}\operatorname{tr}(A_{n_{2}})}=O_{p}\left(1\right)$. 
Specifically, we show in Lemma \ref{lemma: deno} below that  $\left(\frac{n_{2}^{-1}}{1-n_{2}^{-1}\operatorname{tr}\left(A_{n_{2}}\right)}\widehat{d}^{(2)\top}(I-A_{n_{2}})d^{(2)}\right)^{-1}$
is $O\left(1\right)$ and bounded away from 0.
Moreover, we also show below in Lemmas \ref{lemma: T1} and \ref{lemma: T2} that terms $T_{1}$ and $T_{2}$ are $o_{p}\left(1\right)$, respectively. 
Thus,  we have 
\begin{align}
    \widehat{\alpha}  = \alpha +o_p(1).
\end{align}
The proof for the consistency is complete.
\end{proof}

% First we establish the  Lemmas \ref{lemma: deno}, \ref{lemma: T1}, and \ref{lemma: T2}, which are necessary for the theoretical analysis of this paper.

% Before we prove the result in Theorem \ref{thm:asymp normality}, we establish a three preliminary results, Lemmas \ref{lemma: deno}, \ref{lemma: T1}, and \ref{lemma: T2}.

%\subsection{Proof of Lemma 1}
%\label{proofs_lemma1}

\begin{lemma}
    \label{lemma: deno}
   Under  Assumptions \ref{assumption_Eigen_Upsilon}--\ref{assumption_sigma_E}, we have  $\frac{1}{1-n_{2}^{-1}\operatorname{tr}(A_{n_{2}})}n_{2}^{-1}\widehat{d}^{(2)\top}\left(I-A_{n_{2}}\right)d^{(2)}$ is $O(1)$.
\end{lemma}

%\begin{proof}
%    See Appendix \ref{proofs_lemma1}
%\end{proof}

\begin{proof}
First,
following the random matrix theory,
\begin{equation*}
\lim_{n_{2}\rightarrow\infty}n_{2}^{-1}\operatorname{tr}(A_{n_{2}})=\tau\left\{ 1-\eta_{x}m(\eta_{x})\right\}.
\end{equation*}
This result indicates that the expectation of $n^{-1}\operatorname{tr}(A_{n})$ could converge to a non-zero constant.
We recommend the Lemma 2  of  \citet{liu2020estimation} for details of the relevant random matrix results.
% \footnote{See Lemma 2 of \cite{liu2020estimation}.} 

We can write $\vartheta_{X}=X^{(2)}\gamma_{vx}$.
We denote that $P_{z}^{\left(1\right)}=n_{1}^{-1}Z^{(2)}(n_{1}^{-1}Z^{\left(1\right)\top}Z^{\left(1\right)}+\eta_{z}I)^{-1}Z^{\left(1\right)\top}$, then 
\begin{align}
\widehat{d}^{(2)} & =Z^{(2)}\widehat{\gamma}_{z}\\ \nonumber
 & =P_{z}^{\left(1\right)}Z^{\left(1\right)}\gamma_{z}+P_{z}^{\left(1\right)}V^{\left(1\right)}.
\end{align}
Notice that we use the  $P^{(1)}$ to denote that the Ridge
projection is conducted on dataset $\mS_{1}$. The $V^{\left(1\right)}$
represents the first-stage error term in the dataset $\mS_{1}$. The
projection matrix $P_{z}^{(1)}$ is used to project out its correlation
with the instruments $Z$.

Thus, asymptotically we have 
\begin{align}
\widehat{d}^{(2)}=X^{(2)}\gamma_{vx}^{*}+\vartheta_{X^{c}}+P_{z}^{\left(1\right)}V^{\left(1\right)}+o_{p}\left(1\right),
\end{align}
where $\vartheta_{X^{c}}$ represents the part of the optimal instrument
orthogonal to $X$, and there is a remainder term that becomes negligible
in probability as the sample size increases. 
Under Assumption \ref{assumption_Eigen_Upsilon}, 
the optimal instruments are divided into two parts: the first part
is in the column space of confounding variables, and the second part is independent of the confounding variables. 

Moreover, 
\begin{align*}
 & \frac{1}{1-n_{2}^{-1}\operatorname{tr}(A_{n_{2}})}n_{2}^{-1}\widehat{d}^{(2)\top}\left(I-A_{n_{2}}\right)d^{(2)}\\
= & \frac{1}{1-n_{2}^{-1}\operatorname{tr}(A_{n_{2}})}n_{2}^{-1}(\gamma_{vx}^{*\top}X^{(2)\top}+\vartheta_{X^{c}}^{\top}+V^{\left(1\right)\top}P_{z}^{\left(1\right)}+o_{p}\left(1\right))\left(I-A_{n_{2}}\right)\left(Z^{(2)}\gamma_{z}+V^{(2)}\right)\\
= & \frac{1}{1-n_{2}^{-1}\operatorname{tr}(A_{n_{2}})}n_{2}^{-1}\gamma_{vx}^{*\top}X^{(2)\top}\left(I-A_{n_{2}}\right)Z^{(2)}\gamma_{z}+\frac{1}{1-n_{2}^{-1}\operatorname{tr}(A_{n_{2}})}n_{2}^{-1}\gamma_{vx}^{*\top}X^{(2)}\left(I-A_{n_{2}}\right)V^{(2)}\\
 & +\frac{1}{1-n_{2}^{-1}\operatorname{tr}(A_{n_{2}})}n_{2}^{-1}\vartheta_{X^{c}}^{\top}\left(I-A_{n_{2}}\right)Z^{(2)}\gamma_{z}+\frac{1}{1-n_{2}^{-1}\operatorname{tr}(A_{n_{2}})}n_{2}^{-1}\vartheta_{X^{c}}^{\top}\left(I-A_{n_{2}}\right)V^{(2)}\\
 & +\frac{1}{1-n_{2}^{-1}\operatorname{tr}(A_{n_{2}})}n_{2}^{-1}V^{\left(1\right)\top}P_{z}^{\left(1\right)}\left(I-A_{n_{2}}\right)Z^{(2)}\gamma_{z}+\frac{1}{1-n_{2}^{-1}\operatorname{tr}(A_{n_{2}})}n_{2}^{-1}V^{\left(1\right)\top}P_{z}^{\left(1\right)}\left(I-A_{n_{2}}\right)V^{(2)}\\
 & +\frac{1}{1-n_{2}^{-1}\operatorname{tr}(A_{n_{2}})}n_{2}^{-1}\left(I-A_{n_{2}}\right)\left(Z^{(2)}\gamma_{z}+V^{(2)}\right)o_{p}\left(1\right).
\end{align*}

We now analyze each of the six terms of the above display. The first
term 
\begin{align*}
 & \frac{1}{1-n_{2}^{-1}\operatorname{tr}(A_{n_{2}})}n_{2}^{-1}\gamma_{vx}^{*\top}X^{(2)\top}\left(I-A_{n_{2}}\right)Z^{(2)}\gamma_{z}\\
 & =\frac{1}{1-n_{2}^{-1}\operatorname{tr}(A_{n_{2}})}n_{2}^{-1}\gamma_{vx}^{*\top}X^{(2)\top}\left(I-A_{n_{2}}\right)\left(X^{(2)}B+W^{(2)}\right)\gamma_{z}\\
 & =\frac{1}{1-n_{2}^{-1}\operatorname{tr}(A_{n_{2}})}n_{2}^{-1}X^{(2)\top}\left(I-A_{n_{2}}\right)X^{(2)}B\gamma_{z}+n_{2}^{-1}\gamma_{vx}^{*\top}X^{(2)\top}\left(I-A_{n_{2}}\right)W^{(2)}\gamma_{z}\\
 & =o_{p}(1),
\end{align*}
 as $n_{2}^{-1}X^{(2)\top}\left(I-A_{n_{2}}\right)X^{(2)}=o_{p}(1)$,
$n_{2}^{-1}\gamma_{vx}^{*\top}X^{(2)\top}\left(I-A_{n_{2}}\right)W^{(2)}=o_{p}(1)$
element-wise, and $\frac{1}{1-n_{2}^{-1}\operatorname{tr}(A_{n_{2}})}=O_{p}\left(1\right)$.

For the second term, we have that $\E\left[n_{2}^{-1}\gamma_{vx}^{*\top}X^{(2)\top}\left(I-A_{n_{2}}\right)V^{(2)}\right]=0$
due to the exogeneity of $X$. Also,
\begin{align*}
\textnormal{Var}\left(n_{2}^{-1}\gamma_{vx}^{*\top}X^{(2)\top}\left(I-A_{n_{2}}\right)V^{(2)}\right) & =\E\left[n_{2}^{-2}\gamma_{vx}^{*\top}X^{(2)\top}\left(I-A_{n_{2}}\right)V^{(2)}V^{(2)\top}\left(I-A_{n_{2}}\right)^{\top}X^{(2)}\gamma_{vx}^{*}\right]\\
 & =\sigma_{v}^{2}\E\left[n_{2}^{-2}\gamma_{vx}^{*\top}X^{(2)\top}\left(I-A_{n_{2}}\right)\left(I-A_{n_{2}}\right)^{\top}X^{(2)}\gamma_{vx}^{*}\right]\\
 & \leq\frac{1}{n_{2}^{2}}\sigma_{v}^{2}C_{X}C_{A}^{2}n_{2}||\gamma_{vx}^{*}||\\
 & =o\left(1\right).
\end{align*}
With assumptions that $\lambda_{\max}\left(\frac{X^{(2)\top}X^{(2)}}{n_{2}}\right)<C$
and $\lambda_{\min}\left(\frac{X^{(2)\top}X^{(2)}}{n_{2}}\right)>c$
for some constants $C_{X}$ and $\epsilon>0$, we can ensure that
$\|I-A_{n_{2}}\|$ is bounded by some constant $C_{A}$. Additionally,
we need the assumption that $\|\gamma_{vx}^{*}\|<C_{\gamma}$ for
some constant $C_{\gamma}$, as well as the homogeneity of $V^{(2)}$.

Regarding the third term, we have that 
\begin{align*}
 & \frac{1}{1-n_{2}^{-1}\operatorname{tr}(A_{n_{2}})}n_{2}^{-1}\vartheta_{X^{c}}^{\top}\left(I-A_{n_{2}}\right)Z^{(2)}\gamma_{z}\\
= & \frac{1}{1-n_{2}^{-1}\operatorname{tr}(A_{n_{2}})}n_{2}^{-1}\vartheta_{X^{c}}^{\top}\left(I-A_{n_{2}}\right)\left(X^{(2)}B+W^{(2)}\right)\gamma_{z}\\
= & \frac{1}{1-n_{2}^{-1}\operatorname{tr}(A_{n_{2}})}n_{2}^{-1}\vartheta_{X^{c}}^{\top}\left(I-A_{n_{2}}\right)X^{(2)}B\gamma_{z}+\frac{1}{1-n_{2}^{-1}\operatorname{tr}(A_{n_{2}})}n_{2}^{-1}\vartheta_{X^{c}}^{\top}\left(I-A_{n_{2}}\right)W^{(2)}\gamma_{z}\\
= & o_{p}(1)+\Sigma_{vW}\gamma_{z},
\end{align*}
where $\Sigma_{vW}=\frac{1}{1-n_{2}^{-1}\operatorname{tr}(A_{n_{2}})}n_{2}^{-1}\vartheta_{X^{c}}^{\top}\left(I-A_{n_{2}}\right)W^{(2)}$
captures the contribution of $W^{(2)}$ after projecting out the component
aligned with $X^{(2)}$. By ensuring $\vartheta_{X^{c}}$ has sufficient
information and variability, we make sure that the projection $\vartheta_{X^{c}}^{\top}(I-A_{n_{2}})W^{(2)}$
is non-trivial and retains relevant variability. By Assumption \ref{assumption_sigma_E}, $\Sigma_{W}$ should be full rank and sufficiently different from the subspace spanned by $X^{(2)}$, thus $W^{(2)}$ is not completely aligned with $X^{(2)}$, the third term is non-degenerate.

Similar to the second term above, for the fourth term, we have that 
\begin{equation*}
\E\left[\frac{1}{1-n_{2}^{-1}\operatorname{tr}(A_{n_{2}})}n_{2}^{-1}\vartheta_{X^{c}}^{\top}\left(I-A_{n_{2}}\right)V^{(2)}\right]=0,
\end{equation*}
and 
\begin{align*}
\textnormal{Var}\left(n_{2}^{-1}\vartheta_{X^{c}}^{\top}\left(I-A_{n_{2}}\right)V^{(2)}\right) & =\E\left[n_{2}^{-2}\vartheta_{X^{c}}^{\top}\left(I-A_{n_{2}}\right)V^{(2)}V^{(2)\top}\left(I-A_{n_{2}}\right)^{\top}\vartheta_{X^{c}}\right]\\
 & =\sigma_{v}^{2}\E\left[n_{2}^{-2}\vartheta_{X^{c}}^{\top}\left(I-A_{n_{2}}\right)\left(I-A_{n_{2}}\right)^{\top}\vartheta_{X^{c}}\right]\\
 & \leq\sigma_{v}^{2}C_{\vartheta}^{2}C_{A}^{2}n_{2}^{-1}\\
 & =o\left(1\right).
\end{align*}
With the assumptions that $\lambda_{\max}\left(\frac{X^{(2)\top}X^{(2)}}{n_{2}}\right)<C_{X}$
and $\lambda_{\min}\left(\frac{X^{(2)\top}X^{(2)}}{n_{2}}\right)>\epsilon$
for some constants $C_{X}$ and $\epsilon>0$, we can ensure that
$\|I-A_{n_{2}}\|$ is bounded by some constant $C_{A}$. Additionally,
we need the assumption that $\|\vartheta_{X^{c}}\|<C_{\vartheta}$
for some constant $C_{\vartheta}$, as well as the homogeneity of
$V^{(2)}$.

Regarding the fifth term we have that $\frac{1}{1-n_{2}^{-1}\operatorname{tr}(A_{n_{2}})}n_{2}^{-1}V^{\left(1\right)\top}P_{z}^{\left(1\right)}\left(I-A_{n_{2}}\right)Z^{(2)}\gamma_{z}=o_{p}(1)$,
because the sample splitting ensures that $V^{(1)}$ and $Z^{(2)}$ are
independent given the sample splitting. Moreover, $\E\left[\frac{1}{1-n_{2}^{-1}\operatorname{tr}(A_{n_{2}})}n_{2}^{-1}V^{\left(1\right)\top}P_{z}^{\left(1\right)}\left(I-A_{n_{2}}\right)Z^{(2)}\gamma_{z}\right]=0$
and 
\begin{equation*}
\textnormal{Var}\left(n_{2}^{-1}V^{(1)\top}P_{z}^{(1)}(I - A_{n_{2}})Z^{(2)}\gamma_{z}\right) \leq \left( n_{2}^{-1} ||V^{(1)\top}|| ||I - A_{n_{2}}|| ||Z^{(2)}|| ||\gamma_{z}|| \right)^2, 
\end{equation*}
given that  $||\gamma_{z}|| $ is bounded in Assumption \ref{assumption_reg} and 
$||I - A_{n_{2}}||$ is bounded by 1.

%$\textnormal{Var}\left(n_{2}^{-1}V^{\left(1\right)\top}P_{z}^{\left(1\right)}\left(I-A_{n_{2}}\right)Z^{(2)}\gamma_{z}\right)$
%is bounded by $||I-A_{n_{2}}||$ and $||\gamma_{z}||$. 
Finally, for the sixth term, $\frac{1}{1-n_{2}^{-1}\operatorname{tr}(A_{n_{2}})}n_{2}^{-1}V^{\left(1\right)\top}P_{z}^{\left(1\right)}\left(I-A_{n_{2}}\right)V^{(2)}=o_{p}(1)$. As previously, we have
$\E\left[n_{2}^{-1}V^{\left(1\right)\top}P_{z}^{\left(1\right)}\left(I-A_{n_{2}}\right)V^{(2)}\right]=0$
by the independence of the $V^{\left(1\right)}$and $V^{(2)}$, and
\begin{align*}
\textnormal{Var}\left(n_{2}^{-1}V^{(1)\top}P_{z}^{(1)}(I-A_{n_{2}})V^{(2)}\right) & =\E\left[n_{2}^{-2}\left(V^{(1)\top}P_{z}^{(1)}(I-A_{n_{2}})V^{(2)}\right)^{2}\right]\\
 & =\E\left[n_{2}^{-2}\operatorname{tr}\left(V^{(1)\top}P_{z}^{(1)}(I-A_{n_{2}})V^{(2)}V^{(2)\top}(I-A_{n_{2}})^{\top}P_{z}^{(1)\top}V^{(1)}\right)\right]\\
 & =n_{2}^{-2}\sigma_{v}^{2}\E\left[V^{(1)\top}P_{z}^{(1)}(I-A_{n_{2}})(I-A_{n_{2}})^{\top}P_{z}^{(1)\top}V^{(1)}\right]\\
 & \leq n_{2}^{-2}\sigma_{v}^{2}\|P_{z}^{(1)}\|^{2}\|(I-A_{n_{2}})\|^{2}\E\left[V^{(1)\top}V^{(1)}\right]\\
 & \leq n_{2}^{-2}\sigma_{v}^{2}\cdot1\cdot C_{A}^{2}\cdot n_{1}\sigma_{v}^{2}=o\left(1\right).
\end{align*}

Combining all the above results, we obtain the result
\begin{equation*}
\frac{1}{1-n_{2}^{-1}\operatorname{tr}(A_{n_{2}})}n_{2}^{-1}\widehat{d}^{(2)\top}\left(I-A_{n_{2}}\right)d^{(2)}=\Sigma_{vW}\gamma_{z}+o_{p}\left(1\right)>0.
\end{equation*}
\end{proof}

%\subsection{Proof of Lemma 2}\label{proofs_lemma2}
Next, we provide another auxiliary result for establishing Theorem \ref{thm:consistency}.
\begin{lemma}
    \label{lemma: T1}
    If $\eta_{x}$ satisfies $n_{2}^{1/2}\eta_{x}\|\gamma_{x}\|=o(1)$,
then
\begin{equation*}
V_{1}=\widehat{d}^{(2)\top}\left(I-A_{n_{2}}\right)X^{(2)}\gamma_{x}=o_{p}\left(1\right).
\end{equation*}
\end{lemma}
%\begin{proof}
%    See Appendix \ref{proofs_lemma2}
%\end{proof}

\begin{proof}
The bias term $V_{1}=\widehat{d}^{(2)\top}\left(I-A_{n_{2}}\right)X^{(2)}\gamma_{x}$.
Notice that
\begin{align*}
n_{2}^{-1}\widehat{d}^{(2)\top}\left(I-A_{n_{2}}\right)X^{(2)}\gamma_{x} & =n_{2}^{-1}\widehat{d}^{(2)\top}\left(I-n_{2}^{-1}X^{(2)}\left(n_{2}^{-1}X^{(2)}{}^{\top}X^{(2)}+\eta_{x}I\right)^{-1}X^{(2)\top}\right)X^{(2)}\gamma_{x}\\
 & =n_{2}^{-1}\widehat{d}^{(2)\top}X^{(2)}\gamma_{x}-n_{2}^{-1}\widehat{d}^{(2)\top}n_{2}^{-1}X^{(2)}\left(n_{2}^{-1}X^{(2)}{}^{\top}X^{(2)}+\eta_{x}I\right)^{-1}X^{(2)\top}X^{(2)}\gamma_{x}\\
 & =n_{2}^{-1}\widehat{d}^{(2)\top}X^{(2)}\gamma_{x}-n_{2}^{-1}\widehat{d}^{(2)\top}n_{2}^{-1}X^{(2)}\left(n_{2}^{-1}X^{(2)}{}^{\top}X^{(2)}+\eta_{x}I\right)^{-1}\\
 & \left(n_{2}^{-1}X^{(2)\top}X^{(2)}+\eta_{x}I-\eta_{x}I\right)\gamma_{x}\\
 & =n_{2}^{-1}\widehat{d}^{(2)\top}X^{(2)}\gamma_{x}-n_{2}^{-1}\widehat{d}^{(2)\top}X^{(2)}\left(n_{2}^{-1}X^{(2)}{}^{\top}X^{(2)}+\eta_{x}I\right)^{-1}\\
 & \left(n_{2}^{-1}X^{(2)\top}X^{(2)}+\eta_{x}I\right)\gamma_{x}\\
 & +n_{2}^{-1}\widehat{d}^{(2)\top}X^{(2)}\left(n_{2}^{-1}X^{(2)}{}^{\top}X^{(2)}+\eta_{x}I\right)^{-1}\eta_{x}I\gamma_{x}\\
 & =n_{2}^{-1}\widehat{d}^{(2)\top}X^{(2)}\left(n_{2}^{-1}X^{(2)}{}^{\top}X^{(2)}+\eta_{x}I\right)^{-1}\eta_{x}I\gamma_{x}.
\end{align*}
Then, the variance can be written as
\begin{align*}
V_{1} & =\widehat{d}^{(2)\top}\left(I-A_{n_{2}}\right)X^{(2)}\gamma_{x}\\
 & =n_{2}^{-1}\widehat{d}^{(2)\top}X^{(2)}\left(n_{2}^{-1}X^{(2)}{}^{\top}X^{(2)}+\eta_{x}I\right)^{-1}\eta_{x}\gamma_{x}\\
 & =o_{p}(1).
\end{align*}
The last equality holds with the assumption on the tuning parameter
$n_{2}^{1/2}\eta_{x}\|\gamma_{x}\|=o(1)$. 
\end{proof}

%\subsection{Proof of Lemma 3}\label{proofs_lemma3}
Finally, we state and prove another auxiliary result for establishing Theorem \ref{thm:consistency}.

\begin{lemma}
    \label{lemma: T2}
    We denote that $\E_{X}\left[\epsilon^{(2)}\epsilon^{(2)\top}\right]=\sigma_{\epsilon}^{2}$
and $\E_{X}\left[v^{\left(1\right)}v^{\left(1\right)\top}\right]=\sigma_{v}^{2}$. Then, we have that
    \begin{align*}
V_{2} & =\widehat{d}^{(2)\top}\left(1-A_{n_{2}}\right)\epsilon^{(2)} = o_p(1).
\end{align*} 
\end{lemma}
%\begin{proof}
%    See Appendix \ref{proofs_lemma3}
%\end{proof}

\begin{proof}
Let us write $V_{2}$ as following
\begin{align*}
V_{2} & =\widehat{d}^{(2)\top}\left(I-A_{n_{2}}\right)\epsilon^{(2)}\\
 & =\left(P_{z}^{\left(1\right)}Z^{\left(1\right)}\gamma_{z}+P_{z}^{\left(1\right)}V^{\left(1\right)}\right)\left(I-A_{n_{2}}\right)\epsilon^{(2)}.
\end{align*}

Without sample splitting, the term: $V_{2}=\widehat{d}^{(2)\top}(I-A_{n_{2}})\epsilon^{(2)}$
is not $o_{p}(1)$ because of the dependencies between $\widehat{d}^{(2)}$
and $\epsilon^{(2)}$. This lack of independence complicates the expectation
and variance calculations, leading to potentially large and non-diminishing
terms.

By sample splitting and the exogeneity of $Z$ and $X$, we have
\begin{equation*}
\E\left[n_{2}^{-1}\gamma_{z}Z^{\left(1\right)\top}P_{z}^{\left(1\right)\top}\left(I-A_{n_{2}}\right)\epsilon^{(2)}\right]=0.
\end{equation*}

Moreover, $\E\left[n_{2}^{-1}V^{\left(1\right)\top}P_{z}^{\left(1\right)}(I-A_{n})\epsilon^{(2)}\right]=\E\left[\operatorname{tr}(n_{2}^{-1}\left(I-A_{n_{2}}\right)\epsilon^{(2)}V^{\left(1\right)\top})\right]=0$. 

Now we consider
\begin{align}
 & \textnormal{Var}\left(n_{2}^{-1}\widehat{d}^{(2)\top}\left(1-A_{n_{2}}\right)\epsilon^{(2)}\right)\label{consistency_Var_1}\\
 & =n_{2}^{-2}\E\left[\left\{ \widehat{d}^{(2)\top}\left(1-A_{n_{2}}\right)\epsilon^{(2)}\epsilon^{(2)\top}\left(I-A_{n_{2}}\right)\widehat{d}^{(2)}\right\} \right]\nonumber \\
 & =n_{2}^{-2}\sigma_{\epsilon}^{2}\E\left[(P_{z}^{\left(1\right)}V\left(1\right)+P_{z}^{\left(1\right)}Z^{\left(1\right)}\gamma_{z})^{\top}\left(I-A_{n_{2}}\right)^{2}(P_{z}^{\left(1\right)}V^{\left(1\right)}+P_{z}^{\left(1\right)}Z^{\left(1\right)}\gamma_{z})\right]\nonumber \\
 & =\underset{\textnormal{Term1}}{\underbrace{n_{2}^{-2}\sigma_{\epsilon}^{2}\E\left[V^{\left(1\right)}{}^{\top}P_{z}^{\left(1\right)\top}\left(I-A_{n_{2}}\right)^{2}P_{z}^{\left(1\right)}V^{\left(1\right)}\right]}}\nonumber \\
 & +\underset{\textnormal{Term2}}{\underbrace{2n_{2}^{-2}\sigma_{\epsilon}^{2}\E\left[V^{\left(1\right)}{}^{\top}P_{z}^{\left(1\right)\top}\left(I-A_{n_{2}}\right)^{2}P_{z}^{\left(1\right)}Z^{\left(1\right)}\gamma_{z}\right]}}\nonumber \\
 & +\underset{\textnormal{Term3}}{\underbrace{n_{2}^{-2}\sigma_{\epsilon}^{2}\E\left[\gamma_{z}^{\top}Z^{\left(1\right)}{}^{\top}P_{z}^{\left(1\right)\top}\left(I-A_{n_{2}}\right)^{2}P_{z}^{\left(1\right)}Z^{\left(1\right)}\gamma_{z}\right]}}.\nonumber 
\end{align}

Recall that we have $\left(I-A_{n_{2}}\right)^{2}=I-A_{n_{2}}-n_{2}^{-1}\eta_{x}X^{(2)}\left(X^{(2)}{}^{\top}X^{(2)}+\eta_{x}I\right)^{-2}X^{(2)\top}$.
Denote $\ensuremath{H^{\left(1\right)}=\left(n_{1}^{-1}Z^{\left(1\right)\top}Z^{\left(1\right)}+\eta_{z}I\right)^{-1}Z^{\left(1\right)\top}}$,
thus $\ensuremath{P_{z}^{\left(1\right)}=n_{1}^{-1}Z^{\left(1\right)}H^{\left(1\right)}}$
. Consider the three terms respectively,

For Term1, 
\begin{align*} \label{eq_term1}
 & n_{2}^{-2}\sigma_{\epsilon}^{2}\E\left[V^{\left(1\right)}{}^{\top}P_{z}^{\left(1\right)\top}\left(I-A_{n_{2}}\right)^{2}P_{z}^{\left(1\right)}V^{\left(1\right)}\right]\nonumber\\
= & n_{2}^{-2}\sigma_{\epsilon}^{2}\E\left[\operatorname{tr}\left\{ V^{\left(1\right)}{}^{\top}P_{z}^{\left(1\right)\top}\left(I-A_{n_{2}}\right)^{2}P_{z}^{\left(1\right)}V^{\left(1\right)}\right\} \right]\nonumber\\
= & n_{2}^{-2}\sigma_{\epsilon}^{2}\E\left[\operatorname{tr}\left\{ P_{z}^{\left(1\right)\top}\left(I-A_{n_{2}}\right)^{2}P_{z}V^{\left(1\right)}V^{\left(1\right)\top}\right\} \right]\nonumber\\
= & n_{2}^{-2}\sigma_{\epsilon}^{2}\sigma_{v}^{2}\E\left[\operatorname{tr}\left\{ P_{z}^{\left(1\right)\top}\left(I-A_{n_{2}}\right)^{2}P_{z}^{\left(1\right)}\right\} \right]\nonumber\\
= & n_{2}^{-2}\sigma_{\epsilon}^{2}\sigma_{v}^{2}\E\left[\operatorname{tr}\left\{ n_{1}^{-2}H^{\left(1\right)\top}Z^{\left(1\right)^{\top}}\left(I-A_{n_{2}}\right)^{2}Z^{\left(1\right)}H^{\left(1\right)}\right\} \right]\nonumber\\
= & n_{2}^{-2}\sigma_{\epsilon}^{2}\sigma_{v}^{2}\E\left[\operatorname{tr}\left\{ n_{1}^{-2}H^{\left(1\right)\top}\left(X^{\left(1\right)}B+W^{\left(1\right)}\right)^{\top}\left(I-A_{n_{2}}\right)^{2}\left(X^{\left(1\right)}B+W^{\left(1\right)}\right)H^{\left(1\right)}\right\} \right]\nonumber\\
= & \underset{Term11} {\underbrace{n_{2}^{-4}\sigma_{\epsilon}^{2}\sigma_{v}^{2}\E\left[\operatorname{tr}\left\{ H^{\left(1\right)^{\top}}W^{\left(1\right)^{\top}}\left(I-A_{n_{2}}\right)^{2}W^{\left(1\right)}H^{\left(1\right)}\right\} \right]}}\nonumber\\
 & +\underset{Term12} {\underbrace{2n_{2}^{-4}\sigma_{\epsilon}^{2}\sigma_{v}^{2}\E\left[\operatorname{tr}\left\{ H^{\left(1\right)\top}W^{\left(1\right)^{\top}}\left(I-A_{n_{2}}\right)^{2}X^{\left(1\right)}BH^{\left(1\right)}\right\} \right]}}\nonumber\\
 & +\underset{Term13} {\underbrace{n_{2}^{-4}\sigma_{\epsilon}^{2}\sigma_{v}^{2}\E\left[\operatorname{tr}\left\{ H^{\left(1\right)\top}B^{\top}X^{\left(1\right)\top}\left(I-A_{n_{2}}\right)^{2}X^{\left(1\right)}BH^{\left(1\right)}\right\} \right]}}.
\end{align*}

Now we study the three terms in the above display. For Term11, we have 
\begin{align}
 & n_{2}^{-4}\sigma_{\epsilon}^{2}\sigma_{v}^{2}\E_{X}\left[\operatorname{tr}\left\{ H^{\left(1\right)^{\top}}W^{\left(1\right)^{\top}}\left(I-A_{n_{2}}\right)^{2}W^{\left(1\right)}H^{\left(1\right)}\right\} \right]\nonumber \\
= & n_{2}^{-4}\sigma_{\epsilon}^{2}\sigma_{v}^{2}\E_{X}\left[\operatorname{tr}\left\{ W^{\left(1\right)^{\top}}\left(I-A_{n_{2}}\right)^{2}W^{\left(1\right)}H^{\left(1\right)}H^{\left(1\right)^{\top}}\right\} \right]\nonumber \\
= & n_{2}^{-4}\sigma_{\epsilon}^{2}\sigma_{v}^{2}\operatorname{tr}\left\{ \E\left[W^{\left(1\right)^{\top}}\left(I-A_{n_{2}}\right)^{2}W^{\left(1\right)}H^{\left(1\right)}H^{\left(1\right)^{\top}}\right]\right\} \nonumber \\
= & n_{2}^{-4}\sigma_{\epsilon}^{2}\sigma_{v}^{2}\E\left[\operatorname{tr}\left\{ (1-n_{2}^{-1}\operatorname{tr}(A_{n_{2}}))\Sigma_{W}H^{(1)}H^{(1)\top}\right\} \right].\\
= & n_{2}^{-4}\sigma_{\epsilon}^{2}\sigma_{v}^{2}\operatorname{tr}\left((I-A_{n_{2}})^{2}\Sigma_{W}H^{(1)}H^{(1)\top}\right)\nonumber\\
= & o\left(1\right),
\end{align}
 where $\Sigma_{W}$ is the covariance matrix of the $W^{\left(1\right)}$.
The conditional expectation depends on the independent nature of the
data split strategy. Notice that \begin{align}
n_{2}^{-1}\E_{X}\left[W^{\left(1\right)^{\top}}\left(I-A_{n_{2}}\right)^{2}W^{\left(1\right)}\right]=\left\{ 1-n_{2}^{-1}\operatorname{tr}(A_{n_{2}})\right\} \Sigma_{W},\label{consistency_est_SigmaW}
\end{align}
also,
\begin{align*}
\operatorname{tr}(\Sigma_{W}H^{(1)}H^{(1)\top}) & =\operatorname{tr}\left(\Sigma_{W}\left(n_{1}^{-1}Z^{(1)\top}Z^{(1)}+\eta_{z}I\right)^{-1}n_{1}^{-1}Z^{(1)\top}Z^{(1)}\left(n_{1}^{-1}Z^{(1)\top}Z^{(1)}+\eta_{z}I\right)^{-1}\right)\\
 & =O(n_{1}^{-1})= O(n_{2}^{-1}).
\end{align*}

Similarly, the terms $\operatorname{tr}(A_{n_{2}}\Sigma_{W}H^{(1)}H^{(1)\top})$
and $\operatorname{tr}(A_{n_{2}}^{2}\Sigma_{W}H^{(1)}H^{(1)\top})$ are also
$O(n_{2}^{-1})$ given the assumption that $\lambda_{\max}\left(\frac{Z^{(1)\top}Z^{(1)}}{n_{1}}\right)<C_{Z}$, and hence
the entire expression is $o\left(1\right)$. 

Now consider the Term12, 
\begin{align} & n_{2}^{-4}\sigma_{\epsilon}^{2}\sigma_{v}^{2}\E\left[\operatorname{tr}\left\{ H^{\left(1\right)\top}W^{\left(1\right)^{\top}}\left(I-A_{n_{2}}\right)^{2}X^{\left(1\right)}BH^{\left(1\right)}\right\} \right]\nonumber\\
& =  n_{2}^{-3}\sigma_{\epsilon}^{2}\sigma_{v}^{2}\operatorname{tr}\left\{ \E\left[n_{2}^{-1}\E\left(W^{\left(1\right)^{\top}}\left(I-A_{n_{2}}\right)^{2}X^{\left(1\right)}\mid\mS_{1}\right)BH^{\left(1\right)}H^{\left(1\right)\top}\right]\right\} \nonumber\\
& <  \sigma_{\epsilon}^{2}\sigma_{v}^{2}\operatorname{tr}\left\{ \E\left[n_{2}^{-2}\E\left(W^{\left(1\right)^{\top}}\left(I-A_{n_{2}}\right)^{2}X^{\left(1\right)}\mid\mS_{1}\right)\right]\right\} n_{2}^{-2}\operatorname{tr}\left\{ BH^{\left(1\right)}H^{\left(1\right)\top}\right\} \nonumber\\
& =  o(1),
\end{align}
where
\begin{align} & n_{2}^{-2}\E\left[W^{\left(1\right)^{\top}}\left(I-A_{n_{2}}\right)^{2}X^{\left(1\right)}\right]\nonumber\\ \label{consistency_est_SigmaWE}
& =  n_{2}^{-2}\E\left[W^{\left(1\right)^{\top}}\left(I-A_{n_{2}}\right)X^{\left(1\right)}\right]\nonumber\\
 & + n_{2}^{-1}\eta_{x}\E\left[n_{2}^{-1}W^{\left(1\right)^{\top}}X^{(2)}\left(\frac{X^{(2)}{}^{\top}X^{(2)}}{n_{2}}+\eta_{x}I\right)^{-2}n^{-1}X^{(2)\top}X^{\left(1\right)}\right]\nonumber\\
& =  n_{2}^{-2}\E\left[W^{\left(1\right)^{\top}}\left(I-A_{n_{2}}\right)X^{\left(1\right)}\right]+o(1)\nonumber\\
& =  o(1).
\end{align}
The $\E_{X}\left[W^{\left(1\right)^{\top}}\left(I-A_{n_{2}}\right)X^{\left(1\right)}\right]$
converges to zero element-wise. To be precise, we denote $W_{k}$
as the $k-th$ column of the matrix $W^{\left(1\right)}$. Then,
\begin{align*}
\E_{X}\left[W_{k}^{\left(1\right)\top}\left(I-A_{n_{2}}\right)X^{\left(1\right)}\right]&=\E\left[ \operatorname{tr}\left(W_{k}^{\left(1\right)\top}\left(I-A_{n_{2}}\right)X^{\left(1\right)}\right)\right] \\
&=\E\left[\operatorname{tr}\left(\left(I-A_{n_{2}}\right)X^{\left(1\right)}W_{k}^{\left(1\right)\top}\right)\right]=0,
\end{align*}
 due to the independence between $W_{k}^{\left(1\right)\top}$ and
$X^{\left(1\right)}$.

For  Term13 of equation \eqref{eq_term1}, we have that
\begin{align} & n_{2}^{-4}\sigma_{\epsilon}^{2}\sigma_{v}^{2}\E_{X}\left[\operatorname{tr}\left\{ H^{\left(1\right)\top}B^{\top}X^{\left(1\right)\top}\left(I-A_{n_{2}}\right)^{2}X^{\left(1\right)}BH^{\left(1\right)}\right\} \right]\\
& =  n_{2}^{-4}\sigma_{\epsilon}^{2}\sigma_{v}^{2}\E\left[ \operatorname{tr}\left(B^{\top}X^{\left(1\right)\top}\left(I-A_{n_{2}}\right)^{2}X^{\left(1\right)}BH^{\left(1\right)}H^{\left(1\right)\top}\right)\right] \nonumber\\
& =  n_{2}^{-4}\sigma_{\epsilon}^{2}\sigma_{v}^{2}\operatorname{tr}\left\{ \E\left[B^{\top}X^{\left(1\right)\top}\left(I-A_{n_{2}}\right)^{2}X^{\left(1\right)}BH^{\left(1\right)}H^{\left(1\right)\top}\right]\right\} \nonumber\\
& =  \sigma_{\epsilon}^{2}\sigma_{v}^{2}\operatorname{tr}\left\{ \E\left[n_{2}^{-2}\E\left(B^{\top}X^{\left(1\right)\top}\left(I-A_{n_{2}}\right)^{2}X^{\left(1\right)}B\mid\mS_{1}\right)n_{2}^{-2}H^{\left(1\right)}H^{\left(1\right)\top}\right]\right\} \nonumber\\
& <  \sigma_{\epsilon}^{2}\sigma_{v}^{2}o(1)\times \E\left[\operatorname{tr}\left\{ n_{2}^{-2}H^{\left(1\right)}H^{\left(1\right)\top}\right\} \right] \nonumber\\
& =  o(1), \nonumber
\end{align}
where
\begin{align}
\label{consistency_est_XX}
 & n_{2}^{-2}\E\left[X^{\left(1\right)\top}\left(I-A_{n_{2}}\right)^{2}X^{\left(1\right)}\right]\\
& =  n_{2}^{-2}\E\left[B^{\top}X^{\left(1\right)\top}\left(I-A_{n_{2}}\right)^{2}X^{\left(1\right)}B\right] \nonumber \\
& \quad +\eta_{x}n_{2}^{-1}\E\left[B^{\top}X^{\left(1\right)\top}X^{\left(1\right)}n_{2}^{-1}\left(\frac{X^{\left(1\right)\top}X^{\left(1\right)}}{n_{2}} +\eta I\right)^{-2}X^{\left(1\right)\top}X^{\left(1\right)}n_{2}^{-1}B^{\top}\right]\nonumber \\
& <  n_{2}^{-2}\E\left[B^{\top}X^{\left(1\right)\top}\left(I-A_{n_{2}}\right)^{2}X^{\left(1\right)}B\right]+C\eta_{x}n_{2}^{-1}B^{\top}B\nonumber \\
& =  n_{2}^{-2}\E\left[B^{\top}\eta_{x}\left(\frac{X^{\left(1\right)\top}X^{\left(1\right)}}{n_{2}}+\eta_{x}I\right)^{-1}\frac{X^{\left(1\right)\top}X^{\left(1\right)}}{B}\right]+o(1)\nonumber \\
& =  C\eta_{x}n_{2}^{-1}B^{\top}B+o(1)\nonumber \\
& =  o(1).\nonumber 
\end{align}

The last equality holds since we assume that $\|n_{2}^{-1}\eta_{x}B^{\top}B\|=o(1)$
and $\lambda_{max}(\frac{X^{\left(1\right)\top}X^{\left(1\right)}}{n_{2}})<C$
for some constant $C$. This assumption is commonly used bound assumption in the literature, see \cite{he2000parameters}.

Combining the above results we have the Term1 in equation \eqref{consistency_Var_1} converging to zero.

Now we move to the Term 2 of the equation \eqref{consistency_Var_1}, which is
\begin{equation*}
n_{2}^{-2}\sigma_{\epsilon}^{2}\E_{X}\left[v^{\left(1\right)}{}^{\top}P_{z}^{\left(1\right)\top}\left(I-A_{n_{2}}\right)^{2}P_{z}^{\left(1\right)}Z^{\left(1\right)}\gamma_{z}\right]=0.
\end{equation*}

We then study Term 3 of the equation \eqref{consistency_Var_1}. Write it as following
\begin{align}
 & n_{2}^{-2}\sigma_{\epsilon}^{2}\E_{X}\left[\gamma_{z}^{\top}Z^{\left(1\right)\top}P_{z}^{\left(1\right)\top}\left(I_{n}-A_{n}\right)^{2}P_{z}^{\left(1\right)}Z^{\left(1\right)}\gamma_{z}\right]\nonumber \\
& =  n_{2}^{-2}\sigma_{\epsilon}^{2}\E_{X}\Bigl[\gamma_{z}^{\top}Z^{\left(1\right)\top}Z^{\left(1\right)}\left(Z^{\left(1\right)\top}Z^{\left(1\right)}/n_{2}+\eta_{z}I\right)^{-1} \nonumber\\
& \quad Z^{\left(1\right)\top}\left(I-A_{n_{2}}\right)^{2}Z^{\left(1\right)}\left(Z^{\left(1\right)\top}Z^{\left(1\right)}/n_{2}+\eta_{z}I\right)^{-1}Z^{\left(1\right)\top}Z^{\left(1\right)}\gamma_{z}\Bigr]\nonumber \\
& <  C^{2}\sigma_{\epsilon}^{2}n_{2}^{-2}\E_{X}\left[\gamma_{z}^{\top}Z^{\left(1\right)\top}\left(I-A_{n_{2}}\right)^{2}Z^{\left(1\right)}\gamma_{z}\right].
\end{align}

Recall that $H^{\left(1\right)}=(Z^{\left(1\right)}{}^{\top}Z^{\left(1\right)}/n+\eta I)^{-1}Z^{\left(1\right)}{}^{\top}$.
To be precise, we show the calculation below 
\begin{align}
 & n_{2}^{-2}\E_{X}\left[\gamma_{z}^{\top}Z^{\left(1\right)\top}\left(I-A_{n_{2}}\right)^{2}Z^{\left(1\right)}\gamma_{z}\right]\nonumber \\
 & = n_{2}^{-2}\E_{X}\left[\operatorname{tr}\left\{ n_{2}^{-2}Z^{\top}\left(I-A_{n_{2}}\right)^{2}Z^{\left(1\right)}\gamma_{z}\gamma_{z}^{\top}\right\} \right]\nonumber \\
 & = n_{2}^{-2}\E_{X}\left[\operatorname{tr}\left\{ \left(X^{\left(1\right)}B+W^{\left(1\right)}\right)\left(I-A_{n_{2}}\right)^{2}\left(X^{\left(1\right)}B+W^{\left(1\right)}\right)\gamma_{z}\gamma_{z}^{\top}\right\} \right]\nonumber \\
 & = \underset{\textnormal{Term31}}{\underbrace{n_{2}^{-2}\E_{X}\left[\operatorname{tr}\left\{ W^{\left(1\right)\top}\left(I-A_{n_{2}}\right)^{2}W^{\left(1\right)}\gamma_{z}\gamma_{z}^{\top}\right\} \right]}}\nonumber \\
 & +\underset{\textnormal{Term32}}{\underbrace{2n^{-2}\E_{X}\left[\operatorname{tr}\left\{ W^{\left(1\right)\top}\left(I-A_{n_{2}}\right)^{2}X^{\left(1\right)}B\gamma_{z}\gamma_{z}^{\top}\right\} \right]}}\\
 & +\underset{\textnormal{Term33}}{\underbrace{2n_{2}^{-2}\E_{X}\left[\operatorname{tr}\left\{ B^{\top}X^{\left(1\right)\top}\left(I-A_{n_{2}}\right)^{2}X^{\left(1\right)}B\gamma_{z}\gamma_{z}^{\top}\right\} \right]}}.\nonumber 
\end{align}
With the assumption that $\|\gamma_{z}\|<C_\gamma$, we have $\textnormal{Term31}=o(1)$ with
the result in equation \eqref{consistency_est_SigmaW}. Similarly,
$\textnormal{Term32}=o(1)$ with the result in \eqref{consistency_est_SigmaWE}
and $\textnormal{Term33}=o(1)$ with result in equation \eqref{consistency_est_XX}. Hence we proved the result in \eqref{consistency_Var_1} that $\textnormal{Var}\left(n_{2}^{-1}\widehat{d}(I-A_{n_{2}})\epsilon\right)=o(1)$.
\end{proof}

%\subsection{Proof of Lemma 4}
%\label{proofs_lemma4}
\subsection{Proof of Theorem \ref{thm:asymp normality}}

Now we prove Theorem \ref{thm:asymp normality}, which states the asymptotic normality of $\sqrt{n_{2}}\left(\widehat{\alpha}-\alpha\right)$.

\begin{proof}
Consider the following decomposition,
\begin{align}
\sqrt{n_{2}}\left(\widehat{\alpha}-\alpha\right) & =\sqrt{n_{2}}\left(\widehat{d}^{(2)\top}\left(I-A_{n_{2}}\right)d^{(2)}\right)^{-1}\widehat{d}^{(2)\top}\left(I-A_{n_{2}}\right)X^{(2)}\gamma_{x} \notag \\
 & +\sqrt{n_{2}}\left(\widehat{d}^{(2)\top}\left(I-A_{n_{2}}\right)d^{(2)}\right)^{-1}\widehat{d}^{(2)\top}\left(I-A_{n_{2}}\right)\epsilon^{(2)}. \label{eq: asymp linearity}
\end{align}

We now study both terms in the above display. We start by first considering the term, $\widehat{d}^{(2)\top}\left(I-A_{n_{2}}\right)X^{(2)}\gamma_{x}$. Specifically,
\begin{equation*}
% n_{2}^{1/2}n_{2}^{-1}
\frac{1}{\sqrt{n_{2}}}
\widehat{d}^{(2)\top}\left(I-A_{n_{2}}\right)X^{(2)}\gamma_{x}=n_{2}^{1/2}n_{2}^{-1}\widehat{d}^{(2)\top}X^{(2)}\left(\frac{1}{n_{2}}X^{(2)\top}X^{(2)}+\eta I\right)^{-1}\eta_{x}\gamma_{x}=o(1).
\end{equation*}
With the condition $n_{2}^{1/2}\eta_{x}\|\gamma_{x}\|=o(1)$,
% that $\eta_{x}$ satisfies
the bias term could be a high-order term and thus does not impact the asymptotic
distribution.

Next, we consider the term $\frac{1}{n_{2}}\widehat{d}^{(2)\top}(I-A_{n_{2}})d^{(2)}$.
Recall that
\begin{equation*}
\widehat{\gamma}_{z}=\left(\frac{1}{n_{1}}Z^{(1){\top}}Z^{(1)}+\eta_{z}I\right)^{-1}\frac{1}{n_{1}}Z^{(1){\top}}d^{(1)},
\end{equation*}
where $d^{(1)}=Z^{(1)}\gamma_{z}+X^{(1)}\gamma_{x}+\epsilon^{(1)}.$
As $n_{1}\to\infty$, under Assumptions \ref{assumption_Eigen_Upsilon}--\ref{assumption_large_p}, the estimator $\widehat{\gamma}_{z}$ converges
in probability to:
\begin{equation*}
\gamma_{z}^{*}=(\Sigma_{z}+\eta_{z}I)^{-1}\Sigma_{z}\gamma_{z},
\end{equation*}
where $\ensuremath{\Sigma_{z}=\E[Z_{i}Z_{i}^{\top}]}$ is the
covariance matrix of the instrumental variables. Hence, the predicted value in the second stage is:
\begin{equation*}
\widehat{d}^{(2)}=Z^{(2)}\widehat{\gamma}_{z} =  Z^{(2)}\gamma_{z}^{*} + o(1).
\end{equation*}

Next, by expanding the term of interest, we have:
\begin{equation*}
\frac{1}{n_{2}}\widehat{d}^{(2)\top}(I-A_{n_{2}})d^{(2)}=\frac{1}{n_{2}}\gamma_{z}^{*\top}Z^{(2)\top}(I-A_{n_{2}})\left(Z^{(2)}\gamma_{z}+X^{(2)}\gamma_{x}+\epsilon^{(2)}\right).
\end{equation*}
Since $(\ensuremath{I-A_{n_{2}}})$ removes the effect of $X^{(2)}$, that is $(I-A_{n_{2}})X^{(2)}=0$, and $\ensuremath{\epsilon^{(2)}}$ is independent of $Z^{(2)}$, the leading term simplifies to:
\begin{equation*}
\frac{1}{n_{2}}\gamma_{z}^{*\top}Z^{(2)\top}(I-A_{n_{2}})Z^{(2)}\gamma_{z}.
\end{equation*}
As $n_{2}\to\infty$, by the law of large numbers, we obtain:
\begin{equation*}
\frac{1}{n_{2}}Z^{(2)\top}(I-A_{n_{2}})Z^{(2)}\xrightarrow{p}\E\left[Z_{i}^{\top}(I-A)Z_{i}\right],
\end{equation*}
where $A$ is the projection matrix of $X^{(2)}$ in the limit and
we have the independence of $X$ and $Z$,
\begin{equation*}
\E\left[Z_{i}^{\top}Z_{i}\right]=\text{tr}(\Sigma_{z}).
\end{equation*}

Combining this with the first-stage shrinkage estimator $\ensuremath{\gamma_{z}^{*}},$
the probability limit of the denominator is:

\begin{equation*}
\frac{1}{n_{2}}\widehat{d}^{(2)\top}(I-A_{n_{2}})d^{(2)}\xrightarrow{p}
\gamma_{z}^{*\top}\Sigma_{z}\gamma_{z}^{*}=\gamma_{z}^{\top}\Sigma_{z}(\Sigma_{z}+\eta_{z}I)^{-1}\Sigma_{z}(\Sigma_{z}+\eta_{z}I)^{-1}\Sigma_{z}\gamma_{z}.
\end{equation*}

Now we prove the asymptotic normality of the term $\frac{1}{\sqrt{n_{2}}}\widehat{d}^{(2)\top}\left(I-A_{n_{2}}\right)\epsilon^{(2)}$ in equation \eqref{eq: asymp linearity}. We write it as follows: 
\begin{align}\label{eq asymp normality epsilon}
 & n_{2}^{-1/2}\widehat{d}^{(2)\top}\left(I-A_{n_{2}}\right)\epsilon^{(2)}\nonumber \\
& =  \underset{Term1}{\underbrace{n_{2}^{1/2}n_{2}^{-1}\gamma_{z}Z^{\left(1\right)\top}Z^{\left(1\right)}\left(Z^{\left(1\right)\top}Z^{\left(1\right)}/n_{1}+\eta_{z}I\right)^{-1}Z^{(2)\top}\left(I-A_{n_{2}}\right)\epsilon^{(2)}}}\nonumber \\
& +\underset{Term2}{\underbrace{n_{2}^{1/2}n_{2}^{-1}V^{\left(1\right)}Z^{\left(1\right)}\left(Z^{\left(1\right)\top}Z^{\left(1\right)}/n_{1}+\eta_{z}I\right)^{-1}Z^{(2)\top}\left(I-A_{n_{2}}\right)\epsilon^{(2)}}}.
\end{align}
The first term in the above display accounts for the uncertainty introduced by the estimated
parameters $\gamma_{z}$ in the first stage. It involves the product
of $Z^{(2)\top}$ with the projection matrix $\left(\frac{Z^{(1)\top}Z^{(1)}}{n_{1}}+\eta_{z}I\right)^{-1}$,
which represents the uncertainty due to the first-stage estimation.
The second term captures the remaining uncertainty from the second stage
error $\epsilon^{(2)\ensuremath{}}$ after accounting for the projection
on $Z^{(2)\ensuremath{}}$. The second-stage error $\epsilon^{(2)}$
further interacts with the projection matrix $(I-A_{n_{2}})$, introducing
additional uncertainty. 

%With a high-dimensional structure, the dependence between the stages can lead to biased and inconsistent estimates. Data splitting ensures that the estimates from the first-stage are independent of the second-stage errors, thereby preserving the validity of the asymptotic properties.

%\XDchange{AN EXAMPLE CAN BE INSERTED HERE}

Notice that 
\begin{align*}
\textnormal{Cov}\left(Term1,Term2\right) & =\sigma_{\epsilon}^{2}\E\left[\gamma_{z} Z^{\left(1\right)\top} Z^{\left(1\right)}\left(Z^{\left(1\right)\top}Z^{\left(1\right)}/n_{1}+\eta_{z}I\right)^{-1}Z^{\left(1\right)\top}\left(I-A_{n_{2}}\right)\right.\\
 & \ \left.V^{\left(1\right)\top}Z^{\left(1\right)\top}\left(Z^{\left(1\right)\top}Z^{\left(1\right)}/n_{1}+\eta_{z}I\right)^{-1} Z^{\left(1\right)}\left(I-A_{n_{2}}\right)\right]\\
 & \overset{p}{\to}\sigma_{\epsilon}^{2}\E\left[V^{\left(1\right)\top}\right]=o_{p}\left(1\right),
\end{align*}
when the assumption that $\eta_{z}$ satisfies $n_{2}^{1/2}\eta_{z}\|\gamma_{z}\|=o(1)$
holds.

The first term in \eqref{eq asymp normality epsilon},
\begin{equation*}
n_{1}^{1/2}n_{1}^{-1}\gamma_{z} Z^{\left(1\right)\top} Z^{\left(1\right)}\left(Z^{\left(1\right)\top} Z^{\left(1\right)}/n_{1}+\eta_{z}I\right)^{-1} Z^{\left(2\right)\top}\left(I-A_{n_{2}}\right)\epsilon^{(2)},
\end{equation*}
converges to a normal distribution with a zero mean straightforwardly following the Central Limit Theorem.

For notation simplicity, we denote 
\begin{equation*}
Q= Z^{\left(1\right)} \left( Z^{\left(1\right)\top} Z^{\left(1\right)}/n_{1}+\eta_{z} I \right)^{-1} Z^{\left(2\right)\top}\left(I-A_{n_{2}}\right),
\end{equation*}
where $Q$ is a $n\times n$ non-symmetric matrix. Thus, the second term
\begin{align}
n_{1}^{1/2}n_{1}^{-1}V^{\left(1\right)}Z^{\left(1\right)}\left( Z^{\left(1\right)\top} Z^{\left(1\right)}/n_{1}+\eta_{z}I\right)^{-1} Z^{\left(2\right)\top}\left(I-A_{n_{2}}\right)\epsilon^{(2)}=n_{1}^{1/2}n_{1}^{-1}V^{\left(1\right)}Q\epsilon^{(2)}.
\end{align}

Now we study the limiting behavior of the previous equation.
First, due to the independent nature of the data split, we have 
\begin{equation*}
n_{1}^{1/2}n_{1}^{-1}\E\left[V^{\left(1\right)}Q\epsilon^{(2) \top}\right]=n_{1}^{1/2}n_{1}^{-1}\E\left[\operatorname{tr}\left\{ V^{\left(1\right)}Q\epsilon^{(2)}\right\} \right]=n_{1}^{1/2}n_{1}^{-1}\E\left[\operatorname{tr}\left\{ Q\epsilon^{(2)}V^{\left(1\right)}\right\} \right]=0.
\end{equation*}
In addition, due to the independent nature of the data split, we have 
\begin{align}
n_{1}^{1/2}n_{1}^{-1}V^{\left(1\right)}Q\epsilon^{(2)\top} & =\underset{\bbF_{1}}{\underbrace{n_{1}^{1/2}n_{1}^{-1}\sum_{i=1}^{n_{1}}Q_{ii}\epsilon_{i}v_{i}}}\\
 & +\underset{\bbF_{2}}{\underbrace{n_{1}^{1/2}n_{1}^{-1}\sum_{i=1,i}^{n_{1}}\sum_{j=1,j\neq i}^{n_{1}}Q_{ij}\epsilon_{i}v_{j}}}.\nonumber 
\end{align}
 Consider the variance of the $\bbF_{1}$ term,
\begin{align}
\sigma_{\bbF_{1}}^{2}=n_{1}^{-1}\sum_{i=1}^{n_{1}}Q_{ii}^{2}\sigma_{\epsilon}^{2}\sigma_{v}^{2}=n^{-1}\operatorname{tr}(Q\circ Q)\sigma_{\epsilon}^{2}\sigma_{v}^{2}.
\end{align}
From the mean zero and bounded variance assumption, we have that
\begin{align}
\bbF_{1}/\sigma_{\bbF_{1}} \overset{d}{\to} N(0,1).
\end{align}
 It is also clear that 
\begin{equation*}
\textnormal{Cov}(\bbF_{1},\bbF_{2})=0.    
\end{equation*}

Now, consider the variance of the $\bbF_{2}$ term,
\begin{align*}
\sigma_{\bbF_{2}}^2=n_{1}^{-1}\sum_{i=1,i}^{n_{1}}\sum_{j=1,j\neq i}^{n_{1}}Q_{ij}^{2}\sigma_{v}^{2}\sigma_{\epsilon}^{2}=n^{-1}\left\{ \operatorname{tr}(QQ^{\top})-\operatorname{tr}(Q\circ Q)\right\} \sigma_{v}^{2}\sigma_{\epsilon}^{2}.
\end{align*}
Then, we can employ Lemma \ref{lemma:martingale} showing that the term $\bbF_{2}$ converges to a mean zero normal distribution with variance
$\sigma_{\bbF_{2}}^2$. 

Therefore, from the results above, we have that
\begin{align}
\frac{\bbF_{1}+\bbF_{2}}{\sqrt{\sigma_{\bbF_{1}}^{2}+\sigma_{\bbF_{1}}^{2}}} \overset{d}{\to} N(0,1).
\end{align}
Notice that $\bbF_{1}$ and $\bbF_{2}$ are independent, and that
\begin{equation*}
\sigma_{\alpha}^{2}=\textnormal{Var}\left(n_{1}^{1/2}n_{1}^{-1}V^{\left(1\right)}Q\epsilon^{(2)\top}\right)=\textnormal{Var}\left(\bbF_{1}+\bbF_{2}\right)=\sigma_{\bbF_{1}}^{2}+\sigma_{\bbF_{2}}^{2}. 
\end{equation*}

Finally, we return to the representation in equation \eqref{eq: asymp linearity}. Substituting results for the first term, and the asymptotic expression obtained above for the term $\frac{1}{n_{2}}\widehat{d}^{(2)\top}(I-A_{n_{2}})d^{(2)}$, we have that:
\begin{align*}
% & =\sqrt{n_{2}}\cdot\frac{\widehat{d}^{(2)\top}(I-A_{n_{2}})(X^{(2)}\gamma_{x}+\epsilon^{(2)})}{n_{2}\gamma_{z}^{*^{\top}}\Sigma_{z}\gamma_{z}^{*}+o_{p}(n_{2})}.\\
\sqrt{n_{2}}(\widehat{\alpha}-\alpha) & = (\gamma_{z}^{*^{\top}}\Sigma_{z}\gamma_{z}^{*})^{-1}\frac{1}{\sqrt{n_{2}}}\cdot\widehat{d}^{(2)\top}(I-A_{n_{2}})(X^{(2)}\gamma_{x}+\epsilon^{(2)})+o_{p}(1)
\end{align*}

The following term has been shown to be $o_{p}(1)$:
\begin{equation*}
\frac{1}{\sqrt{n_{2}}}\cdot\widehat{d}^{(2)\top}(I-A_{n_{2}})X^{(2)}\gamma_{x}=o_{p}(1).
\end{equation*}

Thus, the main contribution comes from the error term:
\begin{equation*}
\frac{1}{\sqrt{n_{2}}}\cdot\widehat{d}^{(2)\top}(I-A_{n_{2}})\epsilon^{(2)}.
\end{equation*}

Based on the analysis of Term 1 and Term 2 in equation \eqref{eq asymp normality epsilon} above, both converge in distribution
to an independent normal distribution with a variance of order $o_{p}(1)$. Therefore, we have:
\begin{equation*}
\widehat{d}^{(2)\top}(I-A_{n_{2}})\epsilon^{(2)}=O_{p}(\sqrt{n_{2}}).
\end{equation*}

Thus, by collecting the terms, the overall expression satisfies:
\begin{align}
\sqrt{n_{2}}\frac{1}{\sigma_{\alpha}^{2}}\left(\widehat{\alpha}-\alpha\right)\overset{d}{\to}N(0,1).
\end{align}

Finally, the asymptotic normality of the proposed methods holds, where\\
$\sigma_{\alpha}^{2}=\left(\sigma_{\bbF_{1}}^{2}+\sigma_{\bbF_{2}}^{2}\right)/\left(\gamma_{z}^{*^{\top}}\Sigma_{z}\gamma_{z}^{*}\right)^{2}$. 

\end{proof}

We state and prove the following auxiliary result to prove Theorem \ref{thm:asymp normality}.

\begin{lemma}
\label{lemma:martingale}
    Let Assumptions \ref{assumption_Eigen_Upsilon}--\ref{assumption_large_p} hold. Then, $\bbF_{2}=n_{1}^{1/2}n_{1}^{-1}\sum_{i=1,i}^{n_{1}}\sum_{j=1,j\neq i}^{n_{1}}Q_{ij}\epsilon_{i}v_{j}$
is normal.
\end{lemma}

\begin{proof}
Define 
\begin{equation*}
\bbF_{21}=n_{1}^{-1}\sum_{i=1,i}^{n_{1}}\sum_{j=1,j\neq i}^{n_{1}}Q_{ij}\epsilon_{i}v_{j}=\sum_{i=1}^{n}\mR_{i},
\end{equation*}
 where $\mR_{i}=n_{1}^{-1}\sum_{j=1,j\neq i}^{n_{1}}Q_{ij}\epsilon_{i}v_{j}$.

Then, 
\begin{equation*}
\E\left[\mR_{i}^{2} \right]=n_{1}^{-2}\sigma_{\epsilon}^{2}\sigma_{v}^{2}\sum_{j=1,j\neq i}^{n_{1}}Q_{ij}^{2},
\end{equation*}
and
\begin{equation*}
\E\left[\mR_{i}\mR_{j}\right]=0\quad\text{for i }\neq\text{j}.
\end{equation*}

Then, we have 
\begin{align}
\textnormal{Var}\left(\bbF_{21}\right)=n_{1}^{-2}\sigma_{v}^{2}\sigma_{\epsilon}^{2}\sum_{i=1,i}^{n_{1}}\sum_{j=1,j\neq i}^{n_{1}}Q_{ij}^{2}=n_{1}^{-2}\sigma_{v}^{2}\sigma_{\epsilon}^{2}\left\{ \operatorname{tr}(QQ^{\top})-\operatorname{tr}(Q\circ Q)\right\} .
\end{align}

Now we present some martingale-related definitions. 
\begin{align*}
W_{ni}=n^{-1}\sum_{j=1}^{i-1}(v_{i}\epsilon_{j}Q_{ij}+\epsilon_{i}v_{j}Q_{ji}),
\end{align*}
 and the filtration is $\mF_{i}=\left\{ (X,Z,\epsilon_{j},v_{j}),j\leq i\right\}$. Also, we have $\E\left[W_{ni}\mid\mF_{i-1}\right]=0$ for $i=2,\dots,n$. It means
that $\{\sum_{j=2}^{i}W_{nj},\mF_{i}:2\leq i\leq n_{1}\}$ is a zero
mean martingale. Thus, we define $\tW_{n,i}=\E\left[W_{ni}^{2}\mid\mF_{i-1}\right]$
and $\tW_{n}=\sum_{i=2}^{n_{1}}\tW_{ni}$.

It suffices to show by the Central Limit Theorem that 
\begin{align}
\frac{\tW_{n}}{\sigma_{\bbF_{21}}^{2}}\xrightarrow{p}1,
\end{align}
 and for any $\epsilon>0$, 
\begin{align}
\sum_{i=1}^{n_{1}}\sigma_{\bbF_{21}}^{2}\E\left[W_{ni}^{2}I(|W_{ni}|>\epsilon\sigma_{\bbF_{21}})\mid\mF_{i-1}\right]\rightarrow 0.
\end{align}

First, it is easy to see 
\begin{align}
\E\left[W_{ni}^{2}\right]=n_{1}^{-2}\sigma_{v}^{2}\sum_{j=1}^{i-1}\epsilon_{j}^{2}Q_{ij}^{2}+n^{-2}\sigma_{\epsilon}^{2}\sum_{j=1}^{i-1}v_{j}^{2}Q_{ji}^{2}.
\end{align}

Thus, 
\begin{align}
\E\left[\tW_{n}\right]=\E\left[\sum_{i=2}^{n_{1}}\tW_{ni}\right] & =\E\left[n_{1}^{-2}\sigma_{v}^{2}\sum_{i=2}^{n_{1}}\sum_{j=1}^{i-1}\epsilon_{j}^{2}Q_{ji}^{2}\right]+\E\left[n_{1}^{-2}\sigma_{\epsilon}^{2}\sum_{i=2}^{n_{1}}\sum_{j=1}^{i-1}v_{j}^{2}Q_{ij}^{2}\right]\nonumber \\
 & =n_{1}^{-2}\sigma_{\epsilon}^{2}\sigma_{v}^{2}\left\{ \operatorname{tr}(QQ^{\top})-\operatorname{tr}(Q\circ Q)\right\} .
\end{align}
 Now we check the variance of $\frac{\tW_{n}}{\sigma_{\bbF_{21}}^{2}}$
goes to zero.

\begin{align}
\frac{\tW_{n}}{\sigma_{\bbF_{21}}^{2}} & =\underset{\bbK_{1}}{\underbrace{\sigma_{\bbF_{21}}^{-2}n_{1}^{-2}\sigma_{v}^{2}\sum_{i=1}^{n_{1}}\sum_{j=1}^{i-1}\epsilon_{j}^{2}Q_{ij}^{2}+\sigma_{\bbF_{21}}^{-2}n_{1}^{-2}\sigma_{v}^{2}\sum_{i=1}^{n_{1}}\sum_{j=1,j\neq k}^{i-1}\sum_{k=1}^{i-1}\epsilon_{j}\epsilon_{k}Q_{ij}Q_{ik}}}\nonumber \\
 & +\underset{\bbK_{12}}{\underbrace{\sigma_{\bbF_{21}}^{-2}n_{1}^{-2}\sigma_{\epsilon}^{2}\sum_{i=1}^{n_{1}}\sum_{j=1}^{i-1}v_{j}^{2}Q_{ji}^{2}+\sigma_{\bbF_{21}}^{-2}n_{1}^{-2}\sigma_{\epsilon}^{2}\sum_{i=1}^{n_{1}}\sum_{j=1,j\neq k}^{i-1}\sum_{k=1}^{i-1}v_{j}v_{k}Q_{ji}Q_{ki}}}.
\end{align}
 
 Let $\bbK_{1}=\sigma_{\bbF_{21}}^{-2}\sum_{i=2}^{n_{1}}\bbQ_{1i}$, where $\bbQ_{1i}=n_{1}^{-2}\sigma_{v}\sum_{j=1}^{i-1}\epsilon_{j}^{2}Q_{ij}^{2}+n_{1}^{-2}\sigma_{\epsilon}\sum_{j=1}^{i-1}v_{j}^{2}Q_{ji}^{2}$. 

Note that
\begin{align*}
\E\left[\bbQ_{1i}^{2}\right] & =n_{1}^{-4}\sigma_{v}^{4}\sum_{j=1}^{i-1}\epsilon_{j}^{4}Q_{ij}^{4}+n_{1}^{-4}\sigma_{\epsilon}^{4}\sum_{j=1}^{i-1}v_{j}^{4}Q_{ji}^{4}\\
 & +n_{1}^{-4}\sigma_{v}^{4}\sum_{k=1,k\neq j}^{i-1}\sum_{j=1}^{i-1}\epsilon_{j}^{2}\epsilon_{k}^{2}Q_{ij}^{2}Q_{ik}^{2}+n_{1}^{-4}\sigma_{\epsilon}^{4}\sum_{k=1,k\neq j}^{i-1}\sum_{j=1}^{i-1}v_{j}^{2}v_{k}^{2}Q_{ji}^{2}Q_{ki}^{2}\\
 & +2n_{1}^{-4}\sigma_{v}^{2}\sigma_{\epsilon}^{2}\sum_{j=1}^{i-1}\epsilon_{j}^{2}\epsilon_{j}^{2}Q_{ij}^{2}Q_{ji}^{2}+n_{1}^{-4}\sigma_{\epsilon}^{2}\sigma_{v}^{2}\sum_{k=1,k\neq j}^{i-1}\sum_{j=1}^{i-1}\epsilon_{j}^{2}v_{k}^{2}Q_{ij}^{2}Q_{ki}^{2}\\
 & +n_{1}^{-4}\sigma_{\epsilon}^{2}\sigma_{v}^{2}\sum_{k=1,k\neq j}^{i-1}\sum_{j=1}^{i-1}v_{j}^{2}\epsilon_{k}^{2}Q_{ji}^{2}Q_{ik}^{2},
\end{align*}
 where
\begin{equation*}
\bbQ_{1i}=n_{1}^{-2}\sigma_{v}\sum_{j=1}^{i-1}\epsilon_{j}^{2}Q_{ij}^{2}+n_{1}^{-2}\sigma_{\epsilon}\sum_{j=1}^{i-1}v_{j}^{2}Q_{ji}^{2},
\end{equation*}
 and 
\begin{equation*}
\bbQ_{1k}=n_{1}^{-2}\sigma_{v}\sum_{l=1}^{k-1}\epsilon_{l}^{2}Q_{kl}^{2}+n_{1}^{-2}\sigma_{\epsilon}\sum_{l=1}^{k-1}v_{l}^{2}Q_{lk}^{2}.
\end{equation*}

Then, 
\begin{align*}
\E\left[\bbQ_{1i}\bbQ_{1k}\right] & =n_{1}^{-4}\sigma_{v}^{4}\sum_{l=1}^{k-1}\sum_{j=1}^{i-1}\epsilon_{j}^{2}\epsilon_{l}^{2}Q_{ij}^{2}Q_{kl}^{2}+n_{1}^{-4}\sigma_{\epsilon}^{4}\sum_{l=1}^{k-1}\sum_{j=1}^{i-1}v_{j}^{2}v_{l}^{2}Q_{ji}^{2}Q_{lk}^{2}\\
 & +n_{1}^{-4}\sigma_{\epsilon}^{2}\sigma_{v}^{2}\sum_{l=1}^{k-1}\sum_{j=1}^{i-1}\epsilon_{j}^{2}v_{l}^{2}Q_{ij}^{2}Q_{lk}^{2}+n_{1}^{-4}\sigma_{\epsilon}^{2}\sigma_{v}^{2}\sum_{l=1}^{k-1}\sum_{j=1}^{i-1}v_{j}^{2}\epsilon_{l}^{2}Q_{ji}^{2}Q_{kl}^{2}.
\end{align*}

Recall that
\begin{align*}\sigma_{\bbF_{21}}^{2}=n^{-4}\sigma_{v}^{2}\sigma_{\epsilon}^{2}\sum_{i=1,i}^{n_{1}}\sum_{j=1,j\neq i}^{n_{1}}Q_{ij}^{2}=n_{1}^{-2}\sigma_{v}^{2}\sigma_{\epsilon}^{2}\left\{ \operatorname{tr}(QQ^{\top})-\operatorname{tr}(Q\circ Q)\right\} .\end{align*}
Thus, by the assumption stated in the theorem on the matrix $Q$ we have that\footnote{See \cite{chen2010two} for details.}
\begin{align*}
\textnormal{Var}(\bbK_{1}) & =\frac{1}{\sigma_{\bbF_{21}}^{4}}\left\{ \sum_{i=2}^{n_{1}}\E\left[\bbQ_{1i}^{2}\right]+\sum_{k=2,k\neq i}^{n_{1}}\sum_{i=2}^{n_{1}}\E\left[\bbQ_{1i}\bbQ_{1k}\right]-\sigma_{\bbF_{21}}^{2}\right\} \nonumber\\
 & =\frac{1}{\sigma_{\bbF_{21}}^{4}}\left\{ \left(\E\left[\epsilon^{4}\right]-\sigma_{\epsilon}^{4}\right)n_{1}^{-4}\sigma_{v}^{4}\sum_{i=2,j\neq i}^{n_{1}}\sum_{j=1}^{n_{1}}Q_{ij}^{4}+\left(\E\left[\epsilon^{4}\right]-\sigma_{\epsilon}^{4}\right)n_{1}^{-4}\sigma_{v}^{4}\sum_{k=2}^{n_{1}}\sum_{i=2}^{n_{1}}\sum_{j=1,j\neq k}^{n_{1}}Q_{ij}^{2}Q_{kj}^{2}\right.\nonumber\\
 & \left.+\left(\E\left[v^{4}\right]-\sigma_{v}^{4}\right)n_{1}^{-4}\sigma_{\epsilon}^{4}\sum_{i=2,j\neq i}^{n_{1}}\sum_{j=1}^{n_{1}}Q_{ji}^{4}+\left(\E\left[v^{4}\right]-\sigma_{v}^{4}\right)n_{1}^{-4}\sigma_{\epsilon}^{4}\sum_{k=2}^{n_{1}}\sum_{i=2}^{n_{1}}\sum_{j=1,j\neq k}^{n_{1}}Q_{ji}^{2}Q_{jk}^{2}\right\} \nonumber\\
 & \leq C\frac{\operatorname{tr}\{QQ^{\top}QQ^{\top}\}-\operatorname{tr}\{QQ^{\top}\circ QQ^{\top}\}}{\operatorname{tr}(QQ^{\top})-\operatorname{tr}(Q\circ Q)}=o_{p}(1).
\end{align*}

Consider the following term $\bbK_{2}$, 
\begin{align*}
\bbK_{2}=\sigma_{\bbF_{21}}^{-2}n_{1}\sigma_{v}^{2}\sum_{i=1}^{n_{1}}\sum_{j=1,j\neq k}^{i-1}\sum_{k=1}^{i-1}\epsilon_{j}\epsilon_{k}Q_{ij}Q_{ik}+\sigma_{\bbF_{21}}^{-2}n_{1}^{-2}\sigma_{\epsilon}^{2}\sum_{i=1}^{n_{1}}\sum_{j=1,j\neq k}^{i-1}\sum_{k=1}^{i-1}v_{j}v_{k}Q_{ji}Q_{ki}.
\end{align*}

Let us write 
\begin{equation*}
\bbK_{2}=\sigma_{\bbF_{21}}^{-2}\sum_{i=2}^{n_{1}}\bbQ_{2i}
\end{equation*}
 and 
\begin{equation*}
\bbQ_{2i}=n_{1}^{-2}\sum_{k=2}^{i-1}\bbQ_{2ik}
\end{equation*}
 where $\bbQ_{2ik}=\sigma_{v}^{2}Q_{ik}\epsilon_{k}\sum_{j=1}^{k-1}Q_{ij}\epsilon_{j}+\sigma_{\epsilon}^{2}Q_{ki}v_{k}\sum_{j=1}^{k-1}Q_{ji}v_{j}$.

Then, we have 
\begin{align*}
\bbQ_{2ik}^{2} & =\sigma_{v}^{4}Q_{ik}^{2}\epsilon_{k}^{2}(\sum_{j=1}^{k-1}\sum_{l=1}^{k-1}Q_{ij}Q_{il}\epsilon_{j}\epsilon_{l})+\sigma_{\epsilon}^{4}Q_{ki}^{2}v_{k}^{2}(\sum_{j=1}^{k-1}\sum_{l=1}^{k-1}Q_{ji}Q_{li}v_{j}v_{l})\\
 & +2\sigma_{v}^{2}\sigma_{\epsilon}^{2}Q_{ik}Q_{ki}\epsilon_{k}v_{k}(\sum_{j=1}^{k-1}\sum_{l=1}^{k-1}Q_{ij}Q_{li}\epsilon_{j}v_{l}).
\end{align*}
 Note that 
\begin{align*}
\E\left[\bbQ_{2ik}^{2}\right]=\sigma_{v}^{4}\sigma_{\epsilon}^{4}Q_{ik}^{2}(\sum_{j=1}^{k-1}Q_{ij}^{2})+\sigma_{v}^{4}\sigma_{\epsilon}^{4}Q_{ki}^{2}(\sum_{j=1}^{k-1}Q_{ji}^{2})
\end{align*}
 and $\E\left[\bbQ_{2ik}\bbQ_{2is}\right]=0$ for $k\neq s$.

Thus, we have 
\begin{align*}
\E\left[\bbQ_{2i}^{2}\right]=n_{1}^{-2}\sigma_{v}^{4}\sigma_{\epsilon}^{4}\left(\sum_{k=2}^{i-1}\sum_{j=1}^{k-1}Q_{ik}^{2}Q_{ij}^{2}+\sum_{k=2}^{i-1}\sum_{j=1}^{k-1}Q_{ki}^{2}Q_{ji}^{2}\right).
\end{align*}

We denote that $\bbQ_{2i}=n^{-2}\sum_{s=2}^{i-1}\bbQ_{2is}$, $\bbQ_{2k}=n^{-2}\sum_{t=2}^{k-1}\bbQ_{2kt}$,
and hence
\begin{equation*}
\E\left[\bbQ_{2i}\bbQ_{2k}\right]=n_{1}^{-4}\Bigl\{\sum_{s=2}^{min(i,k)-1}\E\left[\bbQ_{2is}\bbQ_{2ks}\right]+\sum_{s=2,s\neq t}^{i-1}\sum_{t=2}^{k-1}\E\left[\bbQ_{2is}\bbQ_{2kt}\right]\Bigr\}.
\end{equation*}

First, we check the first-term 
\begin{equation*}
\bbQ_{2is}=\sigma_{v}^{2}Q_{is}\epsilon_{s}\sum_{j=1}^{s-1}Q_{ij}\epsilon_{j}+\sigma_{\epsilon}^{2}Q_{si}v_{s}\sum_{j=1}^{s-1}Q_{ji}v_{j},
\end{equation*}
 and 
\begin{equation*}
\bbQ_{2ks}=\sigma_{v}^{2}Q_{ks}\epsilon_{s}\sum_{j=1}^{s-1}Q_{kj}\epsilon_{j}+\sigma_{\epsilon}^{2}Q_{sk}v_{s}\sum_{j=1}^{s-1}Q_{jk}v_{j}.
\end{equation*}
 We can calculate the covariance as 
\begin{align*}
\E\left[\bbQ_{2is}\bbQ_{2ks}\right] & =\sigma_{v}^{4}\sigma_{\epsilon}^{2}Q_{is}Q_{ks}(\sum_{j=1}^{s-1}Q_{ij}\epsilon_{j})(\sum_{j=1}^{s-1}Q_{kj}\epsilon_{j})+\sigma_{\epsilon}^{4}\sigma_{v}^{2}Q_{si}Q_{sk}(\sum_{j=1}^{s-1}Q_{ji}v_{j})(\sum_{j=1}^{s-1}Q_{jk}v_{j})\\
 & =\sigma_{v}^{4}\sigma_{\epsilon}^{4}Q_{is}Q_{ks}(\sum_{j=1}^{s-1}Q_{ij}Q_{kj})+\sigma_{\epsilon}^{4}\sigma_{v}^{4}Q_{si}Q_{sk}(\sum_{j=1}^{s-1}Q_{ji}Q_{jk})\\
 & =\sigma_{v}^{4}\sigma_{\epsilon}^{4}\{Q_{is}Q_{ks}(\sum_{j=1}^{s-1}Q_{ij}Q_{kj})+Q_{si}Q_{sk}(\sum_{j=1}^{s-1}Q_{ji}Q_{jk})\}.
\end{align*}

It is simple to see that 
\begin{align*}
\E\left[\bbQ_{2is}\bbQ_{2kt}\right]=0.
\end{align*}
 Then, we have 
\begin{align*}
\E\left[\bbQ_{2i}\bbQ_{2k}\right] & =n_{1}^{-4}\sigma_{v}^{4}\sigma_{\epsilon}^{4}\Bigl\{\sum_{s=2}^{min(i,k)-1}\bigl\{ Q_{is}Q_{ks}(\sum_{j=1}^{s-1}Q_{ij}Q_{kj})+Q_{si}Q_{sk}(\sum_{j=1}^{s-1}Q_{ji}Q_{jk})\bigr\}\Bigr\}\\
 & =n_{1}^{-4}\sigma_{v}^{4}\sigma_{\epsilon}^{4}\sum_{s=2}^{min(i,k)-1}\sum_{j=1}^{s-1}\bigl\{ Q_{is}Q_{ks}Q_{ij}Q_{kj}+Q_{si}Q_{sk}Q_{ji}Q_{jk}\bigr\}.
\end{align*}

Thus, we have 
\begin{align*}
\textnormal{Var}(\bbK_{2}) & =\sigma_{\bbF_{21}}^{-4}\left\{ \sum_{i=2}^{n_{1}}\E\left[\bbQ_{2i}^{2}\right]+\sum_{i=1,i\neq k}^{n_{1}}\sum_{k=1}^{n_{1}}\E\left[\bbQ_{2i}\bbQ_{2k}\right]\right\} \\
 & =\sigma_{\bbF_{21}}^{-4}\sum_{i=2}^{n_{1}}\sigma_{v}^{4}\sigma_{\epsilon}^{4}\left(\sum_{k=2}^{i-1}\sum_{j=1}^{k-1}Q_{ik}^{2}Q_{ij}^{2}+\sum_{k=2}^{i-1}\sum_{j=1}^{k-1}Q_{ki}^{2}Q_{ji}^{2}\right)\\
 & +\sigma_{\bbF_{21}}^{-4}n_{1}^{-4}\sigma_{v}^{4}\sigma_{\epsilon}^{4}\sum_{s=2}^{min(i,k)-1}\sum_{j=1}^{s-1}\{Q_{is}Q_{ks}Q_{ij}Q_{kj}+Q_{si}Q_{sk}Q_{ji}Q_{jk}\}\\
 & \leq\left\{ \operatorname{tr}(QQ^{\top})-\operatorname{tr}(Q\circ Q)\right\} ^{-1}\\
 & \left(\left(\sum_{k=2}^{i-1}\sum_{j=1}^{k-1}Q_{ik}^{2}Q_{ij}^{2}+\sum_{k=2}^{i-1}\sum_{j=1}^{k-1}Q_{ki}^{2}Q_{ji}^{2}\right)+\sum_{s=2}^{min(i,k)-1}\sum_{j=1}^{s-1}\{Q_{is}Q_{ks}Q_{ij}Q_{kj}+Q_{si}Q_{sk}Q_{ji}Q_{jk}\}\right)\\
 & =\frac{1}{\operatorname{tr}(QQ^{\top})-\operatorname{tr}(Q\circ Q)}\operatorname{tr}(QQ^{\top}\circ QQ^{\top})\\
 & =o_{p}(1).
\end{align*}

To conclude, we obtain $\bbK_{1}\xrightarrow{p}1$ and $\bbK_{2}\xrightarrow{p}0$.

It remains to show that for any $\epsilon>0$, 
\begin{align*}
\sum_{i=1}^{n_{1}}\sigma_{\bbF_{21}}^{2}\E\left[W_{ni}^{2}I(|W_{ni}|>\epsilon\sigma_{\bbF_{21}})\mid\mF_{i-1}\right]\rightarrow 0.
\end{align*}

First, we apply the Markov inequality, for any $\epsilon^{*}$,
\begin{align*}
 P\left\{ \sum_{i=1}^{n_{1}}\sigma_{\bbF_{21}}^{2}\E\left[W_{ni}^{2}I(|W_{ni}|>\epsilon\sigma_{\bbF_{21}})\mid\mF_{i-1}\right]>\epsilon^{*}\right\} 
 & <\frac{1}{\epsilon^{*}\sigma_{\bbF_{21}}^{2}}\E\left[W_{ni}^{2}I(|W_{ni}|>\epsilon\sigma_{\bbF_{21}})\mid\mF_{i-1}\right].
\end{align*}

Then, we have that 
\begin{equation*}
\E\left[W_{ni}^{2}I(|W_{ni}|^{2}>\epsilon^{2}\sigma_{\bbF_{21}}^{2})\mid\mF_{i-1}\right]<\frac{1}{\epsilon^{2}\sigma_{\bbF_{21}}^{2}}\sum_{i=2}^{n_{1}}\E\left[W_{ni}^{4}\right]
\end{equation*}
due to the fact that $\E\left[yI(y>t)\right]=\int_{t}^{\infty}y^{2}\frac{1}{y}p(y)dy<\frac{1}{t}\E\left[y^{2}\right]$.

We also have that 
\begin{align*}
P\left\{ \sum_{i=1}^{n_{1}}\sigma_{\bbF_{21}}^{2}\E\left[W_{ni}^{2}I(|W_{ni}|>\epsilon\sigma_{\bbF_{21}})\mid\mF_{i-1}\right]>\epsilon^{*}\right\}  & <\frac{1}{\epsilon^{2}\epsilon^{*}}\frac{1}{\sigma_{\bbF_{21}}^{4}}\sum_{i=2}^{n_{1}}\E\left[W_{ni}^{4}\right].
\end{align*}

It suffices to verify the term $\sigma_{\bbF_{21}}^{-4}\sum_{i=2}^{n_{1}}\E\left[W_{ni}^{4}\right]$ converges. Consider
\begin{align*}
W_{ni}=n_{1}^{-1}\sum_{j=1}^{i-1}(v_{i}\epsilon_{j}Q_{ij}+\epsilon_{i}v_{j}Q_{ji}),
\end{align*}
and
\begin{align*}
W_{ni}^{2}=n_{1}^{-2}v_{i}^{2}(\sum_{j=1}^{i-1}\epsilon_{j}Q_{ij})^{2}+n_{1}^{-2}\epsilon_{i}^{2}(\sum_{j=1}^{i-1}v_{j}Q_{ji})^{2}+2n_{1}^{-2}\epsilon_{i}v_{i}(\sum_{j=1}^{i-1}v_{j}Q_{ji})(\sum_{j=1}^{i-1}\epsilon_{j}Q_{ij}).
\end{align*}

Thus, 
\begin{align*}
\E\left[W_{ni}^{4}\right] & =n_{1}^{-4}v_{i}^{4}(\sum_{j=1}^{i-1}\epsilon_{j}Q_{ij})^{4}+n^{-4}\epsilon_{i}^{4}(\sum_{j=1}^{i-1}v_{j}Q_{ji})^{4}+4n_{1}^{-4}\epsilon_{i}^{2}v_{i}^{2}(\sum_{j=1}^{i-1}v_{j}Q_{ji})^{2}(\sum_{j=1}^{i-1}\epsilon_{j}Q_{ij})^{2}\\
 & +2n_{1}^{-4}\epsilon_{i}^{2}v_{i}^{2}(\sum_{j=1}^{i-1}\epsilon_{j}Q_{ij})^{2}(\sum_{j=1}^{i-1}v_{j}Q_{ji})^{2}.
\end{align*}

Therefore, 
\begin{align*}
\sigma_{\bbF_{21}}^{-4}\sum_{i=2}^{n_{1}}\E\left[W_{ni}^{4}\right] & =n_{1}^{-4}\sigma_{\bbF_{21}}^{-4}\sum_{i=1}^{n_{1}}\left(v^{4}\epsilon^{4}\sum_{j=1}^{i-1}Q_{ij}^{4}+v^{4}\sigma_{\epsilon}^{4}\sum_{j=1,j\neq k}^{i-1}\sum_{k=1}^{i-1}Q_{ij}^{2}Q_{ik}^{2}\right.\nonumber\\
 & +v^{4}\epsilon^{4}\sum_{j=1}^{i-1}Q_{ji}^{4}+\epsilon^{4}\sigma_{v}^{4}\sum_{j=1,j\neq k}^{i-1}\sum_{k=1}^{i-1}Q_{ji}^{4}Q_{ki}^{4}\nonumber\\
 & +4\sigma_{v}^{4}\sigma_{\epsilon}^{4}\sum_{j=1}^{i-1}Q_{ji}^{2}Q_{ij}^{2}+4\sigma_{v}^{4}\sigma_{\epsilon}^{4}\sum_{j=1,j\neq k}^{i-1}\sum_{k=1}^{i-1}Q_{ji}^{2}Q_{ki}^{2}\nonumber\\
 & \left.+ 2\sigma_{v}^{4}\sigma_{\epsilon}^{4}\sum_{j=1}^{i-1}Q_{ji}^{2}Q_{ij}^{2}+2\sigma_{v}^{4}\sigma_{\epsilon}^{4}\sum_{j=1,j\neq k}^{i-1}\sum_{k=1}^{i-1}Q_{ji}^{2}Q_{ki}^{2} \right)\\
 & \leq C\frac{\operatorname{tr}(QQ^{\top}\circ QQ^{\top})}{\operatorname{tr}(QQ^{\top})-\operatorname{tr}(Q\circ Q)^{2}}=o_{p}(1).
\end{align*}
\end{proof}

\subsection{Proof of Theorem \ref{thm:variance}}
\label{proofs_theorem3}

\begin{proof}
Notice that $\sigma_{\alpha}^{2}=\left(d^{\top}\left(I-A_{n_{2}}\right)d\right)^{-1}d^{\top}\left(I-A_{n_{2}}\right)d\left(d^{\top}\left(I-A_{n_{2}}\right)d\right)^{-1}\sigma_{\epsilon}^{2}$.

It remains for us to discuss the estimation of $\sigma_{\epsilon}^{2}$
below. Write 
\[
\widehat{\sigma}_{\epsilon}^{2}=\frac{\check{\sigma}_{\epsilon}^{2}}{1-n_{2}^{-1}\operatorname{tr}\left(A_{n_{2}}\right)},
\]
 where $\check{\sigma}_{\epsilon}^{2}=\text{\ensuremath{\frac{1}{n_{2}}}}Y^{(2)\top}\left(I-P_{n_{2}}\right)Y^{(2)}$.

Denote $S_{n_{2}}=\left[d^{(2)},X^{(2)}\right]$ and $P_{n_{2}}$
is the projection matrix of $S_{n_{2}}$. Thus, we have that $P_{n_{2}}=S_{n_{2}}\left(S_{n_{2}}^{\top}S_{n_{2}}+\eta_{s}I\right)^{-1}S_{n_{2}}^{\top}$.
Now we establish the statistical properties of $\check{\sigma}_{\epsilon}^{2}$.
Let 
\begin{align*}
\check{\sigma}_{\epsilon}^{2} & =\text{\ensuremath{\frac{1}{n_{2}}}}Y^{(2)\top}\left(I-P_{n_{2}}\right)Y^{(2)}\\
 & =\frac{1}{n_{2}}\left(\alpha d^{(2)}+X^{(2)}\gamma_{x}+\epsilon^{(2)}\right)^{\top}\left(I-P_{n_{2}}\right)\left(\alpha d^{(2)}+X^{(2)}\gamma_{x}+\epsilon^{(2)}\right)\\
 & =\underset{A_{1}}{\underbrace{\frac{1}{n_{2}}\alpha^{2}d^{(2)\top}\left(I-P_{n_{2}}\right)d^{(2)}}}+\underset{A_{2}}{2\underbrace{\frac{1}{n_{2}}\alpha d^{(2)\top}\left(I-P_{n_{2}}\right)X^{(2)}\gamma_{x}}}+\underset{A_{3}}{2\underbrace{\frac{1}{n_{2}}\alpha d^{(2)\top}\left(I-P_{n_{2}}\right)\epsilon^{(2)}}}\\
 & \ +\underset{A_{4}}{\underbrace{\frac{1}{n_{2}}\gamma_{x}^{\top}X^{(2)\top}\left(I-P_{n_{2}}\right)X^{(2)}\gamma_{x}}}+\underset{A_{5}}{2\underbrace{\frac{1}{n_{2}}\gamma_{x}^{\top}X^{(2)\top}\left(I-P_{n_{2}}\right)\epsilon^{(2)}}}+\underset{A_{6}}{\underbrace{n_{2}^{-1}\epsilon^{(2)\top}\left(I-P_{n_{2}}\right)\epsilon^{(2)}}}.
\end{align*}

For the term $A_{1}$, 
\begin{align*}
A_{1} & =\frac{1}{n_{2}}\alpha^{2}d^{(2)\top}\left(I-P_{n_{2}}\right)d^{(2)}\\
 & =\frac{1}{n_{2}}\alpha^{2}\left(Z^{(2)}\gamma_{z}+V^{(2)}\right)^{\top}\left(I-P_{n_{2}}\right)\left(Z^{(2)}\gamma_{z}+V^{(2)}\right)
\end{align*}

The first term, $\frac{1}{n_{2}}\alpha^{2}\gamma_{z}^{\top}Z^{(2)\top}(I-P_{n_{2}})Z^{(2)}\gamma_{z}$,
converges to zero in probability because $d^{(2)}$ is included in
the projection matrix $P_{n_{2}}$, making $(I-P_{n_{2}})d^{(2)}$
approximately zero. The second term, $\frac{2}{n_{2}}\alpha^{2}V^{(2)\top}(I-P_{n_{2}})Z^{(2)}\gamma_{z}$,
is $o_{p}(1)$ by the CLT under the standard assumption of independence
between $V^{(2)}$ and $Z^{(2)}$, as $\frac{1}{n_{2}}V^{(2)\top}Z^{(2)}=O_{p}(n_{2}^{-1/2})$.
The third term, $\frac{1}{n_{2}}\alpha^{2}V^{(2)\top}(I-P_{n_{2}})V^{(2)}$,
is also $o_{p}(1)$ by the Law of LLN, as $(I-P_{n_{2}})$ is idempotent
and bounded.

Following a similar analysis as above, we examine $A_{2}=\frac{1}{n_{2}}\alpha d^{(2)\top}\left(I-P_{n_{2}}\right)X^{(2)}\gamma_{x}$,
where $P_{n_{2}}$ is the projection matrix of $[d^{(2)},X^{(2)}]$.
Substituting $d^{(2)}=Z^{(2)}\gamma_{z}+V^{(2)}$ yields $\frac{1}{n_{2}}\alpha(Z^{(2)}\gamma_{z}+V^{(2)})^{\top}(I-P_{n_{2}})X^{(2)}\gamma_{x}$.

As in the previous analysis, this expression decomposes into two terms.
The first term, $\frac{1}{n_{2}}\alpha\gamma_{z}^{\top}Z^{(2)\top}(I-P_{n_{2}})X^{(2)}\gamma_{x}$,
is $o_{p}(1)$ since $X^{(2)}$ is included in the projection matrix.
The second term, $\frac{1}{n_{2}}\alpha V^{(2)\top}(I-P_{n_{2}})X^{(2)}\gamma_{x}$,
is also $o_{p}(1)$ by the same arguments as before, utilizing the
Central Limit Theorem and the properties of the projection matrix.
Therefore, the entire expression converges to zero in probability.

For the term $A_{3}$, the calculation are similar to those for the
term $A_{2}$. Specifically,
\begin{align*}
A_{3} & =\frac{1}{n_{2}}\alpha d^{(2)\top}\left(I-P_{n_{2}}\right)\epsilon^{(2)}\\
 & =\frac{1}{n_{2}}\alpha\left(Z^{(2)}\gamma_{z}+V^{(2)}\right)^{\top}\left(I-P_{n_{2}}\right)\epsilon^{(2)}\\
 & =\frac{1}{n_{2}}\alpha\gamma_{z}^{\top}Z^{(2)\top}\left(I-P_{n_{2}}\right)\epsilon^{(2)}+\frac{1}{n_{2}}\alpha V^{(2)\top}\left(I-P_{n_{2}}\right)\epsilon^{(2)}.
\end{align*}
 Following the same arguments as above, both terms are $o_{p}(1)$,
making$A_{3}$ converge to zero in probability.

For the term $A_{4}$, we analyze $\frac{1}{n_{2}}\gamma_{x}^{\top}X^{(2)\top}(I-P_{n_{2}})X^{(2)}\gamma_{x}$,
where $P_{n_{2}}$ is the projection matrix of $[d^{(2)},X^{(2)}]$.
Following similar reasoning as above, since $X^{(2)}$ is included
in the projection matrix $P_{n_{2}}$, we have $(I-P_{n_{2}})X^{(2)}$
approximately zero. Therefore, this quadratic form converges to zero
in probability, making $A_{4}$ to be $o_{p}(1)$.

Regarding $A_{5},$ we have$\E_{X}\left[A_{5}\right]=0$, and 
\begin{align*}
\textnormal{Var}\left(A_{5}\right) & =4n_{2}^{-2}\E\left\{ \gamma_{x}^{\top}X^{(2)\top}\left(I-A_{n_{2}}\right)\epsilon\epsilon^{\top}\left(I-A_{n_{2}}\right)X^{(2)}\gamma_{x}\right\} \\
 & =4n_{2}^{-2}\sigma_{\epsilon}^{2}\gamma_{x}^{\top}X^{(2)\top}\left(I-A_{n_{2}}\right)^{2}X^{(2)}\gamma_{x}\\
 & =4n_{2}^{-2}\sigma_{\epsilon}^{2}\gamma_{x}^{\top}\left(X^{(2)\top}X^{(2)}-2X^{(2)\top}A_{n_{2}}X^{(2)}+X^{(2)\top}A_{n_{2}}^{2}X^{(2)}\right)\gamma_{x}\\
 & =4n_{2}^{-2}\sigma_{\epsilon}^{2}\gamma_{x}^{\top}\left[X^{(2)\top}X^{(2)}-2X^{(2)\top}X^{(2)}\left\{ I-\eta_{x}\left(\frac{X^{(2)\top}X^{(2)}}{n_{2}}+\eta_{x}I\right)^{-1}\right\} \right.\\
 & \left.+X^{(2)\top}X^{(2)}\left\{ I-\eta_{x}\left(\frac{X^{(2)\top}X^{(2)}}{n_{2}}+\eta_{x}I\right)^{-1}\right\} ^{2}\right]\gamma_{x}\\
 & =4n_{2}^{-2}\eta_{x}^{2}\sigma_{\epsilon}^{2}\gamma_{x}^{\top}X^{(2)\top}X^{(2)}\left(\frac{X^{(2)\top}X^{(2)}}{n_{2}}+\eta_{x}I\right)^{-2}\gamma_{x}\\
 & =4n_{2}^{-2}\eta_{x}^{2}\sigma_{\epsilon}^{2}\left\{ \gamma_{x}^{\top}\left(\frac{X^{(2)\top}X^{(2)}}{n_{2}}+\eta_{x}I\right)^{-1}\gamma_{x}-\eta_{x}\gamma_{x}^{\top}\left(\frac{X^{(2)\top}X^{(2)}}{n_{2}}+\eta_{x}I\right)^{-2}\gamma_{x}\right\} \\
 & \leq4n_{2}^{-2}\eta_{x}^{2}\sigma_{\epsilon}^{2}\left\{ \lambda_{\min}\left(\frac{X^{(2)\top}X^{(2)}}{n_{2}}\right)+\eta_{x}\right\} ^{-1}\|\gamma_{x}\|^{2}\\
 & \leq4n_{2}^{-2}\eta_{x}^{2}\|\gamma_{x}\|^{2}\sigma^{2}\\
 & =o(1),
\end{align*}
 where $\lambda_{\min}\left(\frac{X^{(2)\top}X^{(2)}}{n_{2}}\right)$
is the minimum eigenvalue of $\frac{X^{(2)\top}X^{(2)}}{n_{2}}$,
and $\E_{X}$ is the conditional expectation variance on $X$. If
$n_{2}^{-1}\eta_{x}\|\gamma_{x}\|^{2}=o(1)$, then 
\[
\mathrm{A}_{5}=o_{p}(1).
\]

Let $\gamma_{n}=\epsilon^{(2)\top}A_{n_{2}}\epsilon^{(2)}/\epsilon^{(2)\top}\epsilon^{(2)}$.
Then we have $\E_{X}\left(\epsilon^{(2)\top}A_{n_{2}}\epsilon^{(2)}\right)=\sigma_{\epsilon}^{2}\operatorname{tr}\left(A_{n_{2}}\right)$
and 
\[
1-\gamma_{n}=1-n_{2}^{-1}\operatorname{tr}\left(A_{n_{2}}\right)\left\{ 1+O_{p}\left(n_{2}^{-1/2}\right)\right\} =1-n_{2}^{-1}\operatorname{tr}\left(A_{n_{2}}\right)+O_{p}\left(n_{2}^{-1/2}\right)
\]
 due to the fact that $n_{2}^{-1}\operatorname{tr}\left(A_{n_{2}}\right)$
converges to $\tau\{1-\eta m(\eta)\}$ in probability by Lemma 2 of
the Supplementary Material in \citet{liu2020estimation}.
Thus, it follows that 
\[
A_{6}=\left(1-\gamma_{n}\right)\epsilon^{(2)\top}\epsilon^{(2)}/n_{2}=\left\{ 1-n_{2}^{-1}\operatorname{tr}\left(A_{n_{2}}\right)\right\} \epsilon^{(2)\top}\epsilon^{(2)}/n_{2}+O_{p}\left(n_{2}^{-1/2}\right).
\]

In conclusion, we have 
\[
\widehat{\sigma}_{\epsilon}^{2}=\frac{\check{\sigma}_{\epsilon}^{2}}{1-n_{2}^{-1}\operatorname{tr}\left(A_{n_{2}}\right)}=\frac{A_{6}}{1-n_{2}^{-1}\operatorname{tr}\left(A_{n_{2}}\right)}+o_{p}\left(1\right)=\sigma_{\epsilon}^{2}+o_{p}\left(1\right).
\]
 and that concludes the proof.
\end{proof}

\newpage

\bibliographystyle{chicago}
\bibliography{biblio}

\end{document}